\documentclass[reprint,twocolumn,aps,amsmath,amssymb,pra,superscriptaddress,floatfix,tightenlines,nofootinbib,nobibnotes]{revtex4-2}

\usepackage{graphicx}
\usepackage[colorlinks,linkcolor=blue,urlcolor=blue, citecolor=blue]{hyperref}
\usepackage{xurl}
\usepackage[whole]{bxcjkjatype}
\usepackage{booktabs}
\usepackage{graphicx}
\usepackage{comment}

\usepackage{dsfont}
\usepackage{physics}
\usepackage{mathtools}
\usepackage{amsmath, amssymb, amsfonts, amsthm}
\usepackage{thmtools, thm-restate}
\usepackage[sanserif,full]{complexity}

\usepackage{algorithmic}
\usepackage{algorithm}

\newtheorem{theorem}{Theorem}

\newtheorem{lemma}[theorem]{Lemma}
\newtheorem{definition}[theorem]{Definition}

\DeclareMathOperator{\size}{size}
\DeclareMathOperator{\error}{error}

\DeclareMathOperator{\pr}{\mathbf{Pr}}
\DeclareMathOperator{\CC}{\textsf{CC}}
\DeclareMathOperator{\CQ}{\textsf{CQ}}
\DeclareMathOperator{\QC}{\textsf{QC}}
\DeclareMathOperator{\QQ}{\textsf{QQ}}

\DeclareMathOperator{\HeurC}{\textsf{HeurFBPP/poly}}
\DeclareMathOperator{\HeurQ}{\textsf{HeurFBQP}}
\DeclareMathOperator{\heurC}{\textsf{HeurFBPP}}

\allowdisplaybreaks

\begin{document}

\title{Advantage of Quantum Machine Learning from General Computational Advantages}

\author{Hayata Yamasaki*}
\email{hayata.yamasaki@gmail.com} 
\affiliation{Department of Physics, Graduate School of Science, The Univerisity of Tokyo, 7-3-1 Hongo, Bunkyo-ku, Tokyo, 113-0033, Japan}
\author{Natsuto Isogai*}
\affiliation{Department of Physics, Graduate School of Science, The Univerisity of Tokyo, 7-3-1 Hongo, Bunkyo-ku, Tokyo, 113-0033, Japan}
\author{Mio Murao}
\affiliation{Department of Physics, Graduate School of Science, The Univerisity of Tokyo, 7-3-1 Hongo, Bunkyo-ku, Tokyo, 113-0033, Japan}
\collaboration{Hayata Yamasaki and Natsuto Isogai contributed equally to this work.}

\begin{abstract}
An overarching milestone of quantum machine learning (QML) is to demonstrate the advantage of QML over all possible classical learning methods in accelerating a common type of learning task as represented by supervised learning with classical data.
However, the provable advantages of QML in supervised learning have been known so far only for the learning tasks designed for using the advantage of specific quantum algorithms, i.e., Shor's algorithms.
Here we explicitly construct an unprecedentedly broader family of supervised learning tasks with classical data to offer the provable advantage of QML based on general quantum computational advantages, progressing beyond Shor's algorithms.
Our learning task is feasibly achievable by executing a general class of functions that can be computed efficiently in polynomial time for a large fraction of inputs by arbitrary quantum algorithms but not by any classical algorithm.
We prove the hardness of achieving this learning task for any possible polynomial-time classical learning method. 
We also clarify protocols for preparing the classical data to demonstrate this learning task in experiments.
These results open routes to exploit a variety of quantum advantages in computing functions for the experimental demonstration of the advantage of QML\@.
\end{abstract}

\maketitle

\textit{Introduction.}---
Machine learning technologies supervised by big data serve as one of the core infrastructures to support our daily lives.
Quantum machine learning (QML) attracts growing attention as an emerging field of research to further accelerate and scale up the learning by taking advantage of quantum computation~\cite{wittek2014,schuld2021machine}.
Quantum computation is believed to achieve significant speedup in solving various computational problems over conventional classical computation~\cite{N4,arora2009computational}.
The central goal of supervised learning is, however, not solving the computational problems themselves but finding and making a correct prediction on unseen data under the supervision of given sample data~\cite{kearns1990computational,kearns1994introduction, scholkopf2002learning,bach2021learning}.
A far-reaching milestone in the field of QML is to demonstrate the advantage of QML, i.e., an end-to-end acceleration in accomplishing this goal of learning in such a way that any possible classical learning method would never be able to achieve.

However, it has been challenging to realize this milestone due to our limited theoretical understanding of the learning tasks with the advantage of QML\@.
Representative QML algorithms such as those in Refs.~\cite{rebentrost2014quantum,kerenidis_et_al:LIPIcs.ITCS.2017.49,PhysRevA.99.052331,NEURIPS2020_9ddb9dd5,https://doi.org/10.48550/arxiv.2106.09028,10.5555/3618408.3620034} have theoretically guaranteed upper bounds of their runtime, and it is indeed hard for existing classical algorithms to achieve the same learning tasks as these QML algorithms within a comparable runtime; nevertheless, these facts are insufficient to provably rule out the possibility of the potential existence of classical learning methods achieving the comparable runtime.
For example, the quantum algorithm for recommendation systems was initially claimed to achieve an exponential speedup compared to the existing classical algorithms at the time~\cite{kerenidis_et_al:LIPIcs.ITCS.2017.49}, but it turned out in later research that the quantum algorithm achieves only a polynomial speedup compared to the best possible classical algorithm due to a breakthrough in designing a quantum-inspired classical algorithm for solving the same task~\cite{10.1145/3313276.3316310}.
So far, the advantage of QML in accelerating supervised learning with classical data has been proven only based on the quantum computational advantage of Shor's algorithms~\cite{shor1994algorithms,doi:10.1137/S0097539795293172,shor1999polynomial} to solve integer factoring and discrete logarithms~\cite{servedio2004equivalences, liu2021rigorous}.
The existing techniques to prove the hardness of learning for all possible classical methods use a cryptographic argument essentially depending on the specific mathematical structure of discrete logarithms and integer factoring~\cite{kearns1994Jan,kearns1994introduction,servedio2004equivalences,liu2021rigorous,kearns1990computational}, which do not straightforwardly generalize.
More recently, necessary conditions for constructing learning tasks with the advantage of QML have also been investigated in Refs.~\cite{gyurik2023establishing, gyurik2023exponential}, but these works do not explicitly provide the learning tasks to meet these requirements in general.
A fundamental open question in the theory of QML has been what types of quantum computational advantages, beyond that of Shor's algorithms, lead to learning tasks to demonstrate the end-to-end acceleration of the learning; to address this question, novel techniques need to be established to prove the classical hardness of learning beyond the realm of Shor's algorithms.

Also from a practical perspective, toward the experimental demonstration of the advantages of QML, Shor's algorithms are challenging to realize with near-term quantum technologies~\cite{Gidney2021howtofactorbit}, confronting as an obstacle to the demonstration.
One reason for this practical challenge is rooted in the fact that Shor's algorithms need to compete with the well-established classical algorithms for integer factoring that run only within subexponential time, which is much shorter than exponential time~\cite{Lenstra1992ART,10.1007/BFb0091539,9740707}.
For this reason, a milestone in realizing Shor's algorithms is often set to factorize a relatively large integer, such as a 2048-bit integer~\cite{Gidney2021howtofactorbit}.
On the other hand, apart from Shor's algorithms, polynomial-time quantum algorithms can also solve other types of computational problems that are potentially harder for classical algorithms, such as those relevant to topological data analysis (TDA)~\cite{lloyd2016quantum, hayakawa2022quantum, 9996768, akhalwaya2022towards, mcardle2022streamlined}, Pell's equation~\cite{10.1145/509907.510001, hallgren2007polynomial}, and \textsf{BQP}-complete problems~\cite{freedman2002modular,wocjan2006several, aharonov2006polynomial,PhysRevLett.103.150502,Aharonov_2011,10.1145/3519935.3519991,gharibian2022improved}.
For the classically hard problems, one could potentially use a much smaller size of the problem instance, e.g., with much less than $2048$-bit inputs, to demonstrate the quantum computational advantages.
In view of this, the solution to the above open question on the relation between quantum computational advantages and the advantage of QML will also constitute a significant step to the practical demonstration as well as the fundamental understanding of QML\@.

In this work, we address this open question by showing that quantum advantages in computing functions \textit{in general} lead to the provable end-to-end advantage of QML in conducting supervised learning with classical data, without specifically depending on Shor's algorithms.
In particular, for a general class of functions that can be computed efficiently in polynomial time for a large fraction of inputs by quantum algorithms but not by any classical algorithm (even under the supervision of data), we explicitly construct a family of classification tasks in a conventional setting of supervised learning with classical sample data, i.e., in a probably approximately correct (PAC) learning model~\cite{valiant1984theory,kearns1990computational,kearns1994introduction}.
We construct a polynomial-time quantum algorithm for solving our learning task using a polynomial amount of classical sample data.
This quantum algorithm is simply implementable by a variant of the conventional learning method: feature mapping by quantum computation to map the input classical data into the corresponding bit strings representing their features, followed by linear separation by classical computation to find an appropriate hyperplane in the feature space to achieve the classification.
At the same time, we prove that no polynomial-time classical algorithm can accomplish this learning task.
In contrast with the previous work on the advantage of QML using Shor's algorithms~\cite{servedio2004equivalences, liu2021rigorous}, our results lead to unprecedented candidates of quantum algorithms for the explicit demonstration of the advantage of QML in supervised learning.
Furthermore, we provide a protocol for preparing the classical sample data to demonstrate this advantage of QML in the experiments.
These results open vast opportunities for anyone to use a general class of quantum algorithms of their favorite to achieve QML with a provable advantage over any classical learning method, progressing beyond the previous approach specifically depending on Shor's algorithms.

\textit{Formulation of learning tasks.}---
We describe the setting of learning and the formulation of our learning task.
Our analysis is based on a conventional setting of supervised learning, i.e., the PAC learning model~\cite{valiant1984theory,kearns1990computational, kearns1994introduction}.
See Methods on the definition of the PAC learning model.

Following the convention of the PAC learning, we formulate our concept class $\mathcal{C}_N$, i.e., a set of functions $c\in\mathcal{C}_N$ to be learned, which classify an $N$-bit input $x$ coming from a target probability distribution $\mathcal{D}_N$ into binary-labeled categories specified by $c(x)=0$ or $c(x)=1$.
Our formulation is in line with a conventional learning approach based on feature mapping and linear separation (Fig.~\ref{fig:feature and linear}).
In this approach, the learning algorithm first maps an input $x$ to another vector $f(x)$ and subsequently classifies $x$ in the space of $f(x)$.
This mapping $f$ is known as a feature map, transforming $x$ in the input space into the corresponding feature $f(x)$ in the feature space that encapsulates essential information for the classification.
The sets of $x$ satisfying $c(x)=0$ and $c(x)=1$ are mapped into $\mathcal{F}_0$ and $\mathcal{F}_1$ of $f(x)$, respectively.
The feature map here should be designed so that $\mathcal{F}_0$ and $\mathcal{F}_1$ have linear separability, i.e, the property that a hyperplane in the feature space should be able to distinguish between $\mathcal{F}_0$ and $\mathcal{F}_1$~\cite{cover1965geometrical}.
More formally, there should exist a vector $s$ in the feature space and a threshold $t$ such that
\begin{equation}
\label{eq:linear separation in introduction}
\begin{aligned}
    f(x)\cdot s \leq t \text{ for } c(x)=0;\quad
    f(x)\cdot s > t \text{ for } c(x)=1.
\end{aligned}
\end{equation}
The equation $f(x)\cdot s = t$ represents the hyperplane to separate $\mathcal{F}_0$ and $\mathcal{F}_1$ specified by the unknown target concept $c$ to be learned.
The concept class of $c$ is learnable by converting the given input samples using the feature map, followed by finding this hyperplane, i.e., its parameter $s$, using the corresponding output samples.
Once we find $s$, for a new input $x$ drawn from $\mathcal{D}_N$, we can make a correct prediction of $c(x)$ by evaluating a hypothesis $h(x)$ in a hypothesis class, which classifies $x$ based on the value of $f(x)\cdot s$.

\begin{figure}[t]
    \centering
    \includegraphics[width=3.4in]{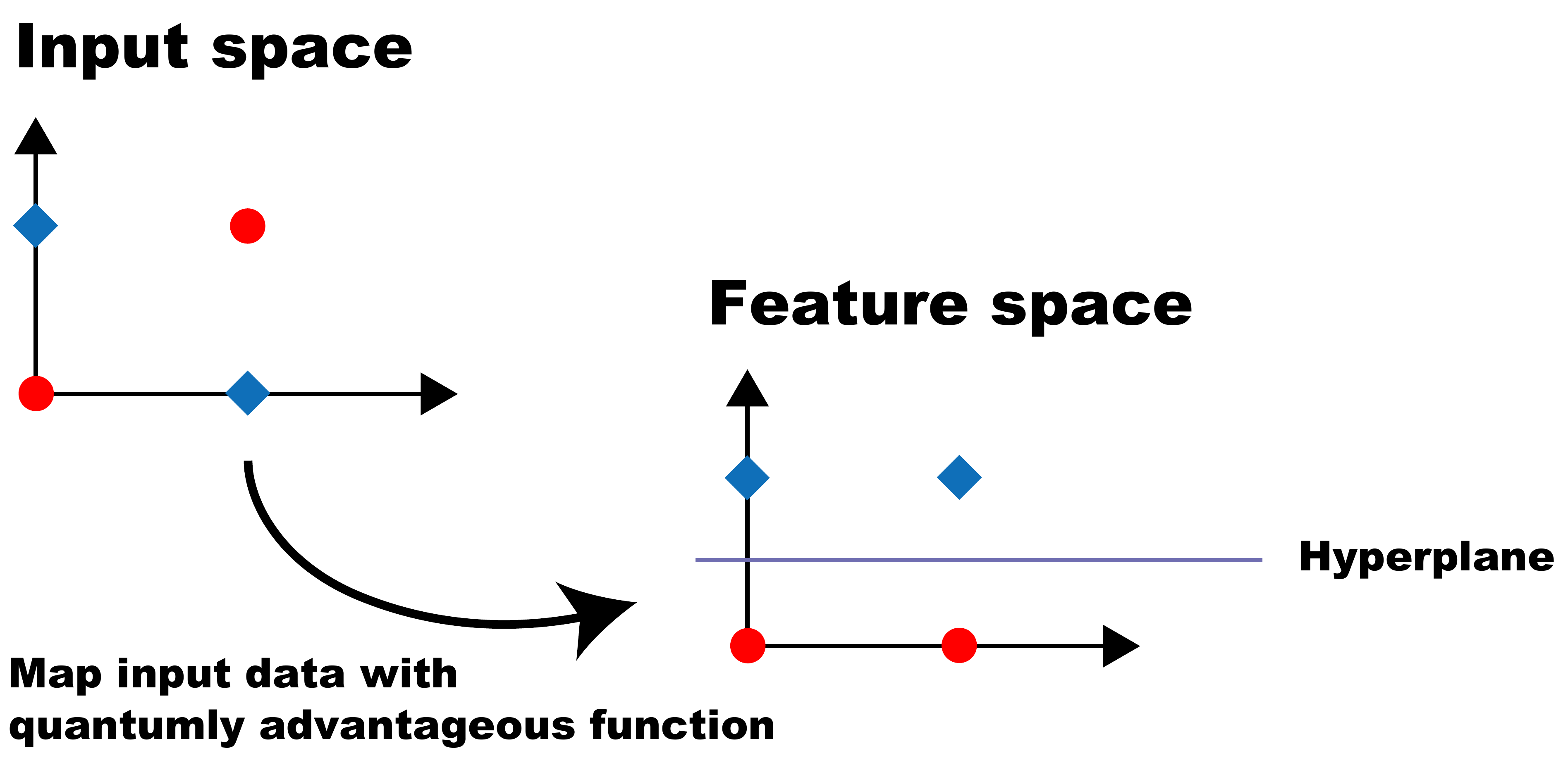}
    \caption{A conventional learning approach based on feature mapping and linear separation, which our work also follows. Inputs $x$ in the input space (red circles with output labels $c(x)=0$ and blue squares with $c(x)=1$) are mapped into the corresponding features $f(x)$ in the feature space by a feature map $f$. Then, using the features $f(x)$ of the input samples and the corresponding output samples $c(x)$, we find a hyperplane linearly separating the sets of features for $c(x)=0$ and $c(x)=1$ as in~\eqref{eq:linear separation in introduction} to achieve the learning. In our learning tasks, we use quantumly advantageous functions as $f$.}
    \label{fig:feature and linear}
\end{figure}

Based on this approach, we construct our concept class $\mathcal{C}_N=\{c_s\}_{s}$ parameterized by $s$ in the vector space $\mathbb{F}_2^D$ over a finite field, where $\mathbb{F}_2=\{0,1\}$ is the finite field representing a bit, and $D$ is the dimension of the feature space  $\mathbb{F}_2^D$.
Each concept $c_s$ is a function from the input space $\{0,1\}^N$ to binary labels $\{0,1\}$.
With some choice of the feature map $f_N:\{0,1\}^N\to\mathbb{F}_2^D$, we here define $c_s$ as
\begin{equation}
\label{eq:concept class in introduction}
    c_s(x) \coloneqq f_N(x) \cdot s \in \mathbb{F}_2 = \{0,1\},
\end{equation}
where $f_N(x) \cdot s$ is the bitwise inner product in  $\mathbb{F}_2^D$.
This concept class is designed in accordance with the convention in machine learning based on feature mapping and linear separation as in~\eqref{eq:linear separation in introduction}, yet using the finite fields as the feature space (i.e., $f_N(x) \cdot s=t\coloneqq 0$ or $f_N(x) \cdot s=1$).

\textit{Advantage of QML from general quantum computational advantages.}---
To seek the advantage of QML, we study our concept class in~\eqref{eq:concept class in introduction} with an appropriate choice of the feature map $f_N$.
Remarkably, for our concept class, we show that $f_N$ can be arbitrarily chosen from, roughly speaking, a general class of functions that can be computed efficiently within a polynomial time by quantum algorithms but not by classical algorithms.
In the following, we introduce this general class of functions for $f_N$, followed by describing how the advantage of QML emerges from this general quantum advantage in computing $f_N$.

Under a target distribution $\mathcal{D}_N$, we choose the feature map $f_N$ to be any function in a general class denoted by
\begin{equation}
\label{eq:quantumly_advantageous}
    \{(f_N,\mathcal{D}_N)\}_{N}\in\HeurQ\setminus(\HeurC),
\end{equation}
as defined more precisely in the following (see also Appendix~\ref{sec:quantum_advantage} for more detail).
We call a function in this class a \textit{quantumly advantageous function} under $\mathcal{D}_N$.

In~\eqref{eq:quantumly_advantageous}, we require that $f_N(x)$ should be efficiently computable by a quantum algorithm for a large fraction of $x$ drawn from $\mathcal{D}_N$, and the set $\HeurQ$ of pairs of the function $f_N$ and the distribution $\mathcal{D}_N$ represents those satisfying this requirement.
In the computational complexity theory, a (worst-case) complexity class $\FBQP$~\cite{10.1007/978-3-642-20712-9_1} conventionally represents a set of functions $f_N(x)$ that is efficiently computable by a quantum algorithm for all $x$ even for the worst-case input~\footnote{Note that Ref.~\cite{10.1007/978-3-642-20712-9_1} defines $\FBQP$ as a class of search problems, i.e., computation of functions having a set of multiple outputs for each input, but we consider $f_N$ to have a single output $f_N(x)$ for each input $x$.}, but for PAC learning, it suffices to work on a \textit{heuristic} complexity class $\HeurQ$ that requires the efficiency not all but only a large fraction of $x$~\cite{TCS-004,Impagliazzo}.
More precisely, $\{(f_N,\mathcal{D}_N)\}_N\in\HeurQ$ requires that there exists a quantum algorithm $\mathcal{A}(x,\mu,\nu)$ that should run, for any $N$ and $0 < \mu,\nu<1$, within a polynomial runtime $O(\poly(N,1/\mu,1/\nu))$ to output $f_N(x)$ correctly for a large fraction $1-\mu$ of inputs $x$ drawn from $\mathcal{D}_N$ with a high probability at least $1-\nu$, i.e.,
\begin{align}
\label{eq:quantum}
\pr_{x\sim\mathcal{D}_N}[\pr[\mathcal{A}(x,\mu,\nu) = f_N(x)]\geq 1-\nu]\geq 1-\mu,
\end{align}
where the inner probability is taken over the randomness of $\mathcal{A}$.

At the same time, in~\eqref{eq:quantumly_advantageous}, we require that $f_N(x)$ should not be efficiently computable by any classical (randomized) algorithm for the large fraction of $x$ even with the help of the sample data, and to meet this requirement, it suffices to consider the relative complement of the set $\HeurC$ as in~\eqref{eq:quantumly_advantageous}.
In the complexity theory, $\heurC$ can be considered to be the classical analog of $\HeurQ$ while $\textsf{FBPP}$ the classical analog of $\FBQP$, indicating that the functions $f_N(x)$ are efficiently computable from the given input $x$ by a classical randomized algorithm.
The previous work on the advantage of QML~\cite{servedio2004equivalences,liu2021rigorous} used a cryptographic argument specifically depending on discrete logarithms and integer factoring to prove the classical hardness of their learning tasks, but we here identify that we can use the heuristic complexity class $\heurC$ to rule out the existence of polynomial-time classical learning algorithms for our learning tasks.
In the setting of PAC learning, the sample data are also initially given in addition to $x$.
The use of the sample data is not necessarily limited to the learning, but the data may also potentially help the classical algorithm to compute $f_N(x)$ itself more efficiently~\cite{huang2021power, gyurik2023establishing, gyurik2023exponential}.
The parameters of hypotheses learned from the data can be seen as a bit string of polynomial length $O(\poly(N,1/\mu,1/\nu))$ that could potentially make the computation more efficient, known as the advice string $\alpha$ in the complexity theory~\cite{arora2009computational}.
To address this issue, we require that $f_N(x)$ should remain hard to compute by any classical algorithm even with an arbitrary polynomial-length advice string $\alpha$.
More precisely, $\{(f_N,\mathcal{D}_N)\}_N\notin\HeurC$ requires that no classical (randomized) algorithm $\mathcal{A}(x,\alpha,\mu,\nu)$ with an advice $\alpha$ of length $O(\poly(N,1/\mu,1/\nu))$ should run, for any $N$ and $0 < \mu,\nu <1$, in a polynomial time $O(\poly(N,1/\mu,1/\nu))$ to output $f_N(x)$ correctly for a large fraction $1-\mu$ of $x$ from $\mathcal{D}_N$ with a high probability at least $1-\nu$, i.e.,
\begin{equation}
    \pr_{x\sim\mathcal{D}_N}[\pr[\mathcal{A}(x,\alpha) = f_N(x)]\geq 1-\nu]\geq 1-\mu,
\end{equation}
where the inner probability is taken over the randomness of $\mathcal{A}$.

Our main result proves that for \textit{any} choice of the quantumly advantageous functions $f_N$ in~\eqref{eq:quantumly_advantageous},  our concept class in~\eqref{eq:concept class in introduction} leads to the advantage of QML\@.
In particular, the main result is summarized as follows.
(See also Methods for the more precise definitions of the efficient learnability of the concept and the efficient evaluatability of the hypothesis.)
\begin{theorem}[Advantage of QML from general computational advantages]
\label{thm:informal}
    Under any target distribution $\mathcal{D}_N$ over $N$-bit inputs, for any quantumly advantageous function $f_N$ under $\mathcal{D}_N$, the concept class $\mathcal{C}_N$ defined in~\eqref{eq:concept class in introduction} with $f_N$ is quantumly efficiently learnable, and for this $\mathcal{C}_N$, we can construct a quantumly efficiently evaluatable hypothesis class. By contrast, $\mathcal{C}_N$ is not classically efficiently learnable by any classically efficiently evaluatable hypothesis class.
\end{theorem}
Importantly, Theorem~\ref{thm:informal} establishes the advantage of QML for any quantumly advantageous function, in contrast with the fact that the existing techniques for proving the advantage of QML~\cite{servedio2004equivalences, liu2021rigorous} were limited to using the advantage of Shor's algorithms. 
In Methods, to prove Theorem~\ref{thm:informal}, we explicitly construct polynomial-time quantum algorithms for learning the concept and evaluating the hypothesis; at the same time, we prove that no classical algorithm can evaluate hypotheses that correctly predict the concepts in our concept class.

For our concept, the quantum algorithms for the learning and the evaluation are implementable by the simple approach of feature mapping and linear separation: in our case, the feature mapping uses the quantum algorithm in~\eqref{eq:quantum}, and the linear separation is performed only by classical computing.
A technical challenge in constructing our algorithms is that the learning algorithm does not necessarily find the true parameter $s$ of the target concept $c_s$ but may output an estimate $\tilde{s}$ with $\tilde{s}\neq s$; nevertheless, our analysis shows that the parameter $\tilde{s}$ learned by our algorithm leads to a correct hypothesis $h_{\tilde{s}}$ satisfying $h_{\tilde{s}}(x)=c_s(x)$ for a large fraction of $x$ with high probability (see Methods for detail).

The feature mapping and the linear separation may also be applicable to some of the previous works on the advantage of QML~\cite{liu2021rigorous, gyurik2023establishing, gyurik2023exponential}, but a more crucial difference arises from the techniques for proving the classical hardness.
For instance, a feature map constructed in Ref.~\cite{liu2021rigorous} used Shor's algorithms to transform an $N$-bit input into a feature in a feature space, which was taken as a space of functions called the reproducing kernel Hilbert space (RKHS) in the kernel method~\cite{scholkopf2002learning,bach2021learning} to show a polynomial-time quantum learning algorithm.
However, the existing techniques for proving the classical hardness of such learning tasks needed to use a cryptographic argument depending on the specific mathematical structure of discrete logarithms and integer factoring~\cite{servedio2004equivalences,liu2021rigorous,kearns1990computational,kearns1994Jan,kearns1994introduction} and do not straightforwardly generalize. 
By contrast, we develop techniques for analyzing our learning task with its feature space formulated as the vector space over a finite field, making it possible to prove the classical hardness for any quantumly advantageous function in general (see Methods for detail).

\textit{New directions of QML based on quantumly advantageous functions.}---
Our results in Theorem~\ref{thm:informal} open unprecedented opportunities for using a variety of quantum algorithms to demonstrate the advantage of QML, progressing beyond Shor's algorithms.
We here propose several promising candidates of such quantum algorithms relevant to the following different areas.

\begin{enumerate}
    \item \textit{Topological data analysis (TDA).}---
    Quantum algorithms for computing an estimation of normalized Betti numbers and other topological invariants~\cite{lloyd2016quantum, hayakawa2022quantum, 9996768, akhalwaya2022towards, mcardle2022streamlined} gather considerable attention due to their potential applications to TDA, an area of data science using mathematical tools on topology. The functions computed by these quantum algorithms are leading candidates of the quantumly advantageous functions since techniques for proving the computational hardness are also known in multiple relevant cases under conventional assumptions in the complexity theory~\cite{scheiblechner2007complexity, edelsbrunner2014computational, cade2021complexity, gyurik2022towards, crichigno2022clique, schmidhuber2022complexity}.
    Classical sample data given in terms of bit strings can also be used as the input to some of these quantum algorithms, without necessarily using oracles for the input to these algorithms; for example, given $N$ input points constituting a Vietoris-Rips (VR) complex, the quantum algorithm in Ref.~\cite{hayakawa2022quantum} computes an approximation of the normalized persistent Betti number with accuracy $O(1/\poly(N))$ with probability at least $1-O(1/\poly(N))$.
    In this case, the function that this quantum algorithm computes can be used as a feature map $f_N:\{0,1\}^{O(\poly(N))} \to \mathbb{F}_2^{O(\log(N))}$, from which we can construct our concept class according to~\eqref{eq:concept class in introduction}.
    Note that the normalized persistent Betti number may have a different value from the (original) persistent Betti number due to the normalization factor, but independently of such mathematical structure, our results lead to the advantage of QML for our concept class.
    \item \textit{Cryptographic problems beyond the scope of Shor's algorithms.}---
    Shor's algorithms~\cite{shor1994algorithms,doi:10.1137/S0097539795293172,shor1999polynomial} stand as polynomial-time quantum algorithms to solve integer factoring and discrete logarithms relevant to Rivest–Shamir–Adleman (RSA) cryptosystem~\cite{10.1145/359340.359342}, and no existing classical algorithm can solve these problems within a polynomial time.
    But it still remains an unsolved open problem whether these are hard to compute for any possible polynomial-time classical algorithm apart from the existing ones.
    If an efficient classical algorithm for solving these problems were found in the future, the previously known advantage of QML depending on Shor's algorithms would also cease to survive.
    Even in such a case, our results open a chance that the advantage of QML can still survive based on another cryptographic problem that can be harder than those solved by Shor's algorithms.
    For example, given an $N$-bit nonsquare positive integer $d$ for Pell's equation $x^2-dy^2=1$, the first $O(\poly(N))$ digits of $\ln(x_1+y_1\sqrt{d})$ for its least positive solution $(x_1,y_1)$ can be computed with a high probability exponentially close to one by a polynomial-time quantum algorithm~\cite{10.1145/509907.510001, hallgren2007polynomial}; at the same time, even if one has a polynomial-time classical algorithm for solving integer factoring and discrete logarithms, it is unknown if one can obtain a polynomial-time classical algorithm for this computation, which is relevant to a key exchange system based on the principal ideal problem~\cite{10.1007/0-387-34805-0_31} (see Refs.~\cite{10.1145/509907.510001, hallgren2007polynomial} for details).
    This computations yields a feature map $f_N:\{0,1\}^{N} \to \mathbb{F}_2^{O(\poly(N))}$, from which we can construct our concept class according to~\eqref{eq:concept class in introduction}.
    \item \textit{\textsf{BQP}-complete problems.}---
    It is indeed a variant of long-standing open problems in the complexity theory to prove the existence of the quantumly advantageous functions unconditionally without any computational hardness assumption; however, a natural candidate for the quantumly advantageous functions is the functions relevant to the hardest problems in the scope of the polynomial-time quantum algorithms, known as \textsf{BQP}-complete problems~\cite{freedman2002modular,wocjan2006several, aharonov2006polynomial,Aharonov_2011,PhysRevLett.103.150502,10.1145/3519935.3519991,gharibian2022improved}.
    For instance, the function used for the local unitary-matrix average eigenvalue (LUAE) problem in Ref.~\cite{wocjan2006several} yields such a candidate.
    Given an $N$-bit string $b$ and a $O(\poly(N))$-bit string representing an $O(\poly(N))$-size quantum circuit to implement a  $2^N\times2^N$ unitary matrix $U$, 
    let $\{\lambda_j\}_j$ denote the set of eigenvalues of $U$ with the corresponding set $\{\ket{\nu_j}\}_j$ of eigenvectors.
    Then, the LUAE problem involves computation of an estimate of the average eigenvalue $\Bar{\lambda} = \sum_{j=1}^{2^N} |\braket{b}{\nu_j}|^2\lambda_j$ up to precision $O(1/\poly(N))$ with probability at least $1-O(1/\poly(N))$~\cite{wocjan2006several}.
    Thus, the function to be computed in the LUAE problem provides a feature map $f_N:\{0,1\}^{O(\poly(N))} \to \mathbb{F}_2^{O(\log(N))}$, from which we can construct our concept class according to~\eqref{eq:concept class in introduction}.
    We remark that this construction is based on the worst-case complexity class $\BQP$, but for our concept class, we can indeed choose $f_N$ based on the hardest problems in the heuristic complexity class $\HeurQ$, which is potentially even broader than the worst-case complexity class.
\end{enumerate}

\textit{Protocol for preparing classical sample data for the experimental demonstration.}---
To embody the opportunities for demonstrating the advantage of QML in experiments, we clarify the protocol for preparing the classical sample data for our concept class $\mathcal{C}_N$ in~\eqref{eq:concept class in introduction}.

\begin{figure}[t]
    \includegraphics[width=3.4in]{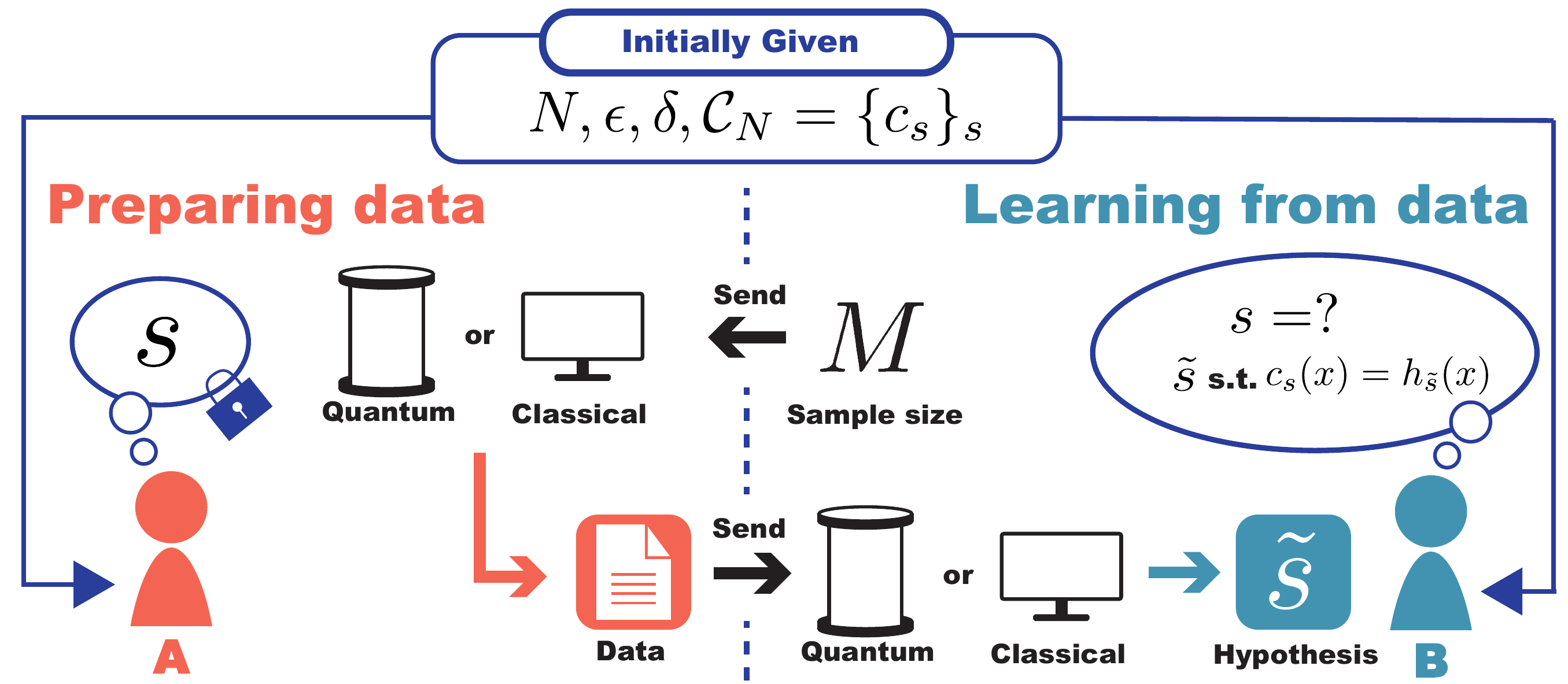}
    \caption{A setup for demonstrating the advantage of QML in supervised learning by two parties $A$ and $B$, where $A$ is in charge of preparing the classical sample data, and $B$ receives the data from $A$ to perform the learning. The parties $A$ and $B$ are initially given the problem size $N$, the error parameter $\epsilon$, the significance parameter $\delta$, and the concept class $\mathcal{C}_N=\{c_s\}_s$ in~\eqref{eq:concept class in introduction} by choosing the feature map as a quantumly advantageous function $f_N$. The party $A$ chooses a $D$-bit parameter $s$ of the target concept $c_s$, which is kept as $A$'s secret. To learn $c_s$, the party $B$ decides the number $M$ of sample data to be used for $B$'s learning and lets $A$ know $M$. Then, $A$ prepares $M$ input-output sample data as described in the main text and sends the data to $B$. Using the given data, $B$ performs the algorithms in Theorem~\ref{thm:informal} to find a $D$-bit string $\tilde{s}$ and make a prediction for new inputs $x$ by the hypothesis $h_{\tilde{s}}(x)=f_N(x)\cdot\tilde{s}$ so that the error in estimating true $c_s(x)$ should be below $\epsilon$ with high probability at least $1-\delta$.}
    \label{fig:enter-label in intro}
\end{figure}

For the demonstration, we consider a two-party setting, where a party $A$ is in charge of preparing the classical sample data, and the other party $B$ receives the data from $A$ to perform the learning (Fig.~\ref{fig:enter-label in intro}). 
Initially, $A$ and $B$ are given the problem size $N$, the error parameter $\epsilon$, and the significance parameter $\delta$.
Let $\mathcal{D}_N$ be a target distribution, and suppose that $A$ can load a sequence of inputs $x$ sampled from $\mathcal{D}_N$ (e.g., from some input data storage), with each $x$ loadable in a unit time.
Note that the exact description of the true distribution $\mathcal{D}_N$ may be unknown to both $A$ and $B$ throughout the learning.
In addition, $A$ and $B$ are given the concept class $\mathcal{C}_N$ determined by choosing the feature map $f_N$ as a quantumly advantageous function under $\mathcal{D}_N$.
Given this initial setup, $B$ decides the number $M=O(\poly(N,1/\epsilon,1/\delta))$ of sample data to be used for $B$'s learning and lets $A$ know $M$ (see Appendix~\ref{sec:framework} for the precise choice of $M$). 
Then, $A$ chooses the parameter $s\in\mathbb{F}_2^D$ of the target concept $c_s\in\mathcal{C}_N$ arbitrarily (e.g., by sampling $s$ uniformly at random). 
For the given $M$ and this choice of $c_s$, $A$ is in charge of preparing $M$ input-output pairs of sample data and giving the data to $B$ while keeping $s$ as $A$'s secret.
The task for $B$ is to find a vector $\tilde{s}\in\mathbb{F}_2^D$ using the $M$ data sent from $A$ and make a prediction for new inputs $x$ drawn from $\mathcal{D}_N$ by the hypothesis $h_{\tilde{s}}(x)=f_N(x)\cdot\tilde{s}$ so that the error in estimating true $c_s(x)$ should be below $\epsilon$ with high probability at least $1-\delta$.

For $A$, we propose a data-preparation protocol using quantum computation in the same way as $B$ using quantum computation for learning, while we will also discuss another protocol using only classical computation later.
To prepare the data, $A$ first loads $M$ inputs $x_1,\ldots,x_M$ drawn from $\mathcal{D}_N$.
Then, using the quantum algorithm~\eqref{eq:quantum} for computing $f_N$, $A$ prepares the corresponding $M$ outputs denoted by $\mathcal{A}(x_1),\ldots,\mathcal{A}(x_M)$.
While the outputs $\mathcal{A}(x)$ may not always be the same as $f_N(x)$ due to the randomness of the quantum algorithm, our analysis shows that, with an appropriate choice of parameters $\mu$ and $\nu$ in~\eqref{eq:quantum}, $A$ can make the error in these $M$ outputs negligibly small by only using a polynomial time of quantum computation (see Appendix~\ref{sec:quantum_data} for detail).
Finally, $A$ send $\{x_m,\mathcal{A}(x_m)\}_{m=1}^M$ to $B$ as the $M$ data.

Note that in the previous works on the advantage of QML based on Shor's algorithms~\cite{servedio2004equivalences, liu2021rigorous}, the data for their learning tasks were able to be prepared by classical computation, but we here put $A$ and $B$ on equal footing by allowing both to compute $f_N$ by quantum computation.
In the previous works, the data were prepared by assuming a special property of cryptographic primitives (i.e., \textit{classically one-way} permutation, which is hard to invert efficiently by classical computation but is invertible efficiently by quantum computation).
In Appendix~\ref{sec:classical_data}, we show that the data preparation for our concept class is also possible using only classical computation if we construct $f_N$ using the classically one-way permutations.

Then, the protocol for achieving $B$'s task reduces to running the quantum learning and evaluation algorithms in Theorem~\ref{thm:informal}.
To compare these quantum algorithms with classical algorithms, we also propose to perform $B$'s task by only using classical computation for several small problem sizes $N$. In particular, from the (superpolynomial) runtimes of classically computing $f_N$ for the several choices of small $N$, we propose to perform extrapolation to estimate the constant factors in the (superpolynomial, potentially exponential) scaling of the runtime of this classical method for larger $N$. 
The demonstration of the advantage of QML is successfully achieved by realizing the polynomial-time quantum algorithms in experiments for an appropriate choice of $N$ to outperform the classical algorithms with the estimated superpolynomial runtime.

\textit{Discussion and outlook.}---
In this work, we have proved that a general class of quantum advantages in computing functions, characterized by $\HeurQ\setminus(\HeurC)$ in~\eqref{eq:quantumly_advantageous}, lead to the end-to-end advantage of QML in a task of supervised learning with classical data.
We have clarified the polynomial-time quantum algorithms to find and make a correct prediction and have also proved the hardness of this learning task for any possible polynomial-time classical method.
Whereas such advantage of QML was shown only for specific cases of using the advantage of Shor's algorithms in previous research~\cite{servedio2004equivalences, liu2021rigorous}, our results make it possible to use the general quantum advantages in computing functions beyond that of Shor's algorithms, such as those relevant to topological data analysis (TDA), Pell's equation, and \textsf{BQP}-complete problems.
Furthermore, we propose protocols to prepare the classical sample data for the experimental demonstration of this advantage of QML\@.
These results solve the fundamental open question about characterizing which types of quantum computational advantages lead to the advantage of QML in accelerating supervised learning, making it possible to exploit all the quantumly advantageous functions for QML\@. 

Our results also constitute a significant step toward the practical demonstration of the advantage of QML in experiments.
For implementation with near-term quantum technologies, heuristic QML algorithms have been studied widely~\cite{mitarai2018quantum, schuld2019quantum, havlivcek2019supervised, PhysRevA.101.032308}, but no analysis provides a classically hard (yet quantumly feasible) instance of the learning tasks; even more problematically, no analysis provides bounds of the runtime and the success probability of these heuristic QML algorithms.
In different settings, advantages of using quantum computation have been shown in a supervised learning setting with quantum data obtained from quantum experiments~\cite{PhysRevLett.126.190505,aharonov2022quantum,chen2022exponential,huang2022quantum,  huang2023learning,meier2023energyconsumption} and also in a distribution learning setting~\cite{liu2018differentiable,sweke2021quantum, pirnay2023superpolynomial}; still, it is not straightforward to apply these quantum algorithms to accelerate the common learning tasks in the era of big data, as represented by supervised learning with classical data.
By contrast, our approach offers a QML framework that can address this type of learning task.
A merit of our QML framework is that it is simply implementable by using any quantumly advantageous function for feature mapping, followed by classically performing linear separation in the feature space, which our framework takes as the space of bit strings representing features.

Also from a broader perspective, in applications of machine learning, state-of-the-art classical learning methods such as deep learning heuristically design the feature maps, e.g., by adapting the architectures of deep neural networks~\cite{Goodfellow-et-al-2016}.
The theoretical analysis of optimized choices of feature maps for given data is challenging even in classical cases, but empirical facts suggest that the classification tasks for data in real-world applications often reduce to applying feature maps designed by such artificial neural networks followed by linear separation~\cite{Goodfellow-et-al-2016}. In view of the success of the artificially designed feature maps, it is crucial to allow as large classes of functions as possible to create more room for the heuristic optimization of the feature maps. 
Advancing ahead, our QML framework makes it possible to design the feature maps flexibly, using arbitrary quantumly advantageous functions to attain the provable advantage of QML\@.
It has also turned out that a large dimension of the feature space is not even a prerequisite for the advantage of QML in our framework, as opposed to a common yet unproven folklore on the potential relevance of large dimension~\cite{mitarai2018quantum, havlivcek2019supervised, schuld2019quantum, PhysRevA.101.032308}; after all, we have proved that the advantage of QML stems solely from the quantum advantages in computing functions without any further requirement for their mathematical structure.
Toward the demonstration of the advantage of QML, an experimental challenge still remains in seeking how to realize quantum computation to surpass the capability of classical computation, and yet our QML framework opens a way to transcend all possible classical learning methods by exploiting \textit{any} realization of quantumly advantageous functions for the feature maps.

\section*{acknowledgements}
This work was supported by JST PRESTO Grant Number JPMJPR201A, JPMJPR23FC, and MEXT Quantum Leap Flagship Program (MEXT QLEAP) JPMXS0118069605, JPMXS0120351339\@.

\section*{Methods}

In Methods, we summarize the probabilistically approximately correct (PAC) learning model and sketch the proof of our main result, i.e., Theorem~\ref{thm:informal} in the main text.

Throughout the paper, we use the following notations.
Let $\mathbb{N}$ be the set of all natural numbers $1,2,\ldots$.
Let $\mathbb{F}_2=\{0,1\}$ denote the finite field of order $2$, i.e., for $0,1\in\mathbb{F}_2$,
\begin{equation}
\begin{aligned}
\label{eq:finite_field}
    0+0=0,\quad 0\times 0=0,\\
    0+1=1,\quad 0\times 1=0,\\
    1+0=1,\quad 1\times 0=0,\\
    1+1=0,\quad 1\times 1=1.
\end{aligned}
\end{equation}
Note that $\{0,1\}$ and $\mathbb{F}_2$ may be the same set, but we will use $\mathbb{F}_2$ for representing the feature space, where we use the bitwise arithmetics in~\eqref{eq:finite_field}, while we will use $\{0,1\}$ in the other cases. 
Let $\poly(x)$ denote a polynomial of $x$.
Unless otherwise stated, we use $N$ as the length of input bit strings for a family of computational problems.
For a probability distribution $\mathcal{D}$ over the set $\mathcal{X}$, we denote by $x\sim\mathcal{D}$ to mean that $x$ is drawn from the distribution $\mathcal{D}$.
A probability $\mathbf{Pr}_{x\sim\mathcal{D}}[\cdots]$ indicates that the probability is taken for the random draw of $x$ according to distribution $\mathcal{D}$.

\textit{PAC learning model.}---
We summarize the definition of the PAC learning model based on Ref.~\cite{kearns1994introduction}.
See also Appendix~\ref{sec:pac_learning} for further details.

In the PAC learning model, for a problem size $N$, a specification of a set of functions $\mathcal{C}_N$, called a concept class, is initially given.
Each function $c\in\mathcal{C}_N$ is called a concept, which maps an $N$-bit input $x$ to a Boolean-valued output $c(x)\in\{0,1\}$ (i.e., a label of $x$ in classification).
For some unknown choice of a concept $c\in\mathcal{C}_N$ called the target concept, the learning algorithm is given a polynomial number of samples $\{x_m,c(x_m)\}_{m=1}^M$, which are pairs of inputs with each $x_m$ drawn from a target probability distribution $\mathcal{D}_N$ and the corresponding outputs $c(x_m)$.
Note that the previous work~\cite{servedio2004equivalences,liu2021rigorous} on the advantage of QML~studied a restricted setting that only allows for a uniform distribution in the choice of the target distribution $\mathcal{D}_N$, but in our work, $\mathcal{D}_N$ can be an arbitrary distribution over the $N$ bits without this restriction.
Using the given sample data, the learning algorithm is designed to find a function, termed a hypothesis $h$, from a set $\mathcal{H}_N$ of functions called a hypothesis class, so as to make a correct prediction on $c$ by $h$.

In the PAC learning model, the ability to find the correct hypothesis for the target concept using the given samples is called the learnability of a concept class under a target distribution~\cite{kearns1994introduction}.
In particular, for the problem size $N$, the error parameter $\epsilon>0$, and a confidence parameter $\delta>0$, a concept class $\mathcal{C}_N$ is \textit{quantumly (classically) efficiently learnable} under $\mathcal{D}_N$ if there exists a quantum (classical randomized) learning algorithm $\mathcal{A}$ that finds a hypothesis $h$ such that
\begin{align}
    \error(h) \coloneqq \pr_{x\sim\mathcal{D}_N}[h(x) \neq c(x)]<\epsilon.
\end{align}
with high probability at least $1-\delta$, using a polynomial number of samples $\{x_m,c(x_m)\}_{i=1}^M$ ($M=O(\poly(N,1/\epsilon,1/\delta))$) within a polynomial time complexity $t_{\mathcal{A}}=O(\poly(N,1/\epsilon,1/\delta))$ (see also Appendix~\ref{Sec: Setting and definition} for more details).

The ability to efficiently evaluate the hypothesis identified from the samples is also crucial in the PAC learning model, which is called evaluatability~\cite{kearns1994introduction}.
By definition of the learnability, the learned hypothesis may inevitably have some error on a nonzero fraction $\epsilon$ of $x$ drawn from $\mathcal{D}_N$, and the definition of evaluatability here also inherits this point.
In particular, for $\epsilon,\delta>0$, we say that a hypothesis class $\mathcal{H}_N$ is \textit{quantumly (classically) efficiently evaluatable} under $\mathcal{D}_N$ if, given a hypothesis $h$, there exists a quantum (classical randomized) evaluation algorithm $\mathcal{A}$ that can compute $h(x)$ for a large fraction $1-\epsilon$ of new inputs $x$ drawn from $\mathcal{D}_N$ with high probability at least $1-\delta$ in terms of the randomness of the (randomized) algorithm $\mathcal{A}$, within a polynomial time complexity $t_{\mathcal{A}}=O(\poly(N,1/\epsilon,1/\delta))$ (see also Appendix~\ref{Sec: Setting and definition} for more details).

\textit{Quantum learning and evaluation algorithms and classical hardness for our concept class.}---
Within the PAC learning model, we sketch the proof of our main result, i.e., Theorem~\ref{thm:informal} in the main text.
In particular, for our concept class, we construct a polynomial-time quantum algorithm for learning the concept to output the corresponding hypothesis in a hypothesis class.
Then, we also construct a polynomial-time quantum algorithm for evaluating the hypothesis in the hypothesis class.
Finally, we prove the hardness of the evaluation of the hypotheses in the hypothesis class for any possible polynomial-time classical algorithm.
See also Appendix~\ref{sec:advantage_QML} for more details.

Our quantum learning algorithm starts with using a quantum algorithm $\mathcal{A}$ in~\eqref{eq:quantum} to compute the feature map $f_N(x_m)$ for each of the given samples $\{(x_m,c_s(x_m))\}_{m=1}^M$.
Note that the features output by $\mathcal{A}$, denoted by $\{\mathcal{A}(x_m)\}_{m=1}^M$, may not exactly coincide with the features $\{f_N(x_m)\}_{m=1}^M$ in general due to the randomness of the quantum algorithm, but our analysis shows that the learning algorithm can feasibly make the failure probability negligibly small.
Our learning algorithm then classically performs Gaussian elimination to solve a system of linear equations for a variable $\tilde{s}\in\mathbb{F}_2^D$
\begin{equation}
\label{eq:linear equation in introduction}
    \begin{aligned}
    \mathcal{A}(x_1)\cdot \tilde{s}&=c_s(x_1),\\
    \mathcal{A}(x_2)\cdot \tilde{s}&=c_s(x_2),\\
    &\vdots\\
    \mathcal{A}(x_M)\cdot \tilde{s}&=c_s(x_M),
\end{aligned}
\end{equation}
subsequently outputting a solution $\tilde{s}$ as an estimate of the parameter of the hypothesis.
For $\tilde{s}$, we construct the hypothesis $h_{\tilde{s}}$ by
\begin{equation}
\label{eq:hypothesis_introduction}
    h_{\tilde{s}}(x) \coloneqq f_N(x)\cdot \tilde{s}.
\end{equation}

A technical challenge in our construction of the learning algorithm arises from the fact that the solutions $\tilde{s}$ of the system of the linear equations in~\eqref{eq:linear equation in introduction} may not be unique. 
After all, we work on a general setting allowing any target distribution $\mathcal{D}_N$, any quantumly advantageous function $f_N$, and any quantum algorithm $\mathcal{A}$ to compute $f_N$ approximately as in~\eqref{eq:quantum}; thus, it may happen that
\begin{equation}
\label{eq:tilde_s_s}
    \tilde{s}\neq s.
\end{equation}
For the worst-case input $x\in\{0,1\}^N$, it may indeed happen that
\begin{equation}
\label{eq:worst_case_input}
    h_{\tilde{s}}(x)\neq c_s(x).
\end{equation}
It is thus nontrivial to prove that the hypothesis $h_{\tilde{s}}$ given by~\eqref{eq:hypothesis_introduction} can predict the target concept $c_s$ correctly as required for the learnability.
We nevertheless prove that any of the solutions $\tilde{s}$ of~\eqref{eq:linear equation in introduction} (even if~\eqref{eq:tilde_s_s} is the case) leads to a correct hypothesis
\begin{equation}
    h_{\tilde{s}}(x)=c_s(x)
\end{equation}
for a large fraction of input $x$ drawn from $\mathcal{D}_N$ with a high probability.
In other words, our analysis proves that the fraction of $x$ causing~\eqref{eq:worst_case_input} can be made negligibly small by our polynomial-time quantum learning algorithm, leading to the quantumly efficient learnability of our concept class (see Appendix~\ref{sec:quantum_algorithm} for details).

As for the quantum algorithm for evaluating the hypothesis, with the parameter $\tilde{s}$ learned, the evaluation algorithm aims to estimate $f_N(x)\cdot \tilde{s}$ in~\eqref{eq:hypothesis_introduction}.
For a new input $x$ drawn from $\mathcal{D}_N$, our evaluation algorithm simply uses the quantum algorithm $\mathcal{A}$ in~\eqref{eq:quantum} to compute $\mathcal{A}(x)$, i.e., an estimate of $f_N(x)$.
The output $\mathcal{A}(x)$ of $\mathcal{A}$ may be different from $f_N(x)$ in general due to the randomness of the quantum algorithm, but we show that the error can be made negligibly small within a polynomial time.
Then, our algorithm takes the (bitwise) inner product of $\mathcal{A}(x)$ and the given parameter $\tilde{s}$, which we prove leads to a correct evaluation of $h(x)$ for a large fraction of input $x$ with a high probability, leading to the quantumly efficient evaluatability (see Appendix~\ref{sec:quantum_algorithm} for details).

Finally, the classical hardness is proved by contradiction, as with the established arguments in the computational learning theory~\cite{kearns1990computational,kearns1994introduction}.
In particular, we prove that if all concepts $c_s(x)$ ($s\in\mathbb{F}_2^D$) of our concept class in~\eqref{eq:concept class in introduction} were classically efficiently learnable by some hypotheses $h_s(x)$ that are classically efficiently evaluatable by polynomial-time classical algorithms, then the feature map $f_N(x)$ in~\eqref{eq:concept class in introduction} would be computed by a polynomial-time classical algorithm using these classical evaluation algorithms, contradicting to the assumption that $f_N$ is a quantumly advantageous function. 
To prove the classical hardness, the previous work on the advantage of QML relied on a cryptographic argument that specifically depends on Shor's algorithms ~\cite{servedio2004equivalences,liu2021rigorous}; by contrast, our proof technique developed here does not depend on such a specific property of Shor's algorithms but is applicable to any quantumly advantageous function in general.

For this development, our key idea is to use the property of the vector space $\mathbb{F}_2^D$ over the finite field used as the feature space in our construction.
In particular, for the standard basis $\{s_d\}_{d=1}^D$ of this $D$-dimensional vector space $\mathbb{F}_2^D$ (i.e., $s_1=(1,0,\ldots,0,0)^\top,\ldots,s_D=(0,0,\ldots,0,1)^\top$),
suppose that the concepts $c_{s_d}(x)$ for $d=1,\ldots,D$ are efficiently learnable by classically efficiently evaluatable hypotheses $h_{s_d}(x)$.
Then, observing that the bitwise inner product $c_{s_d}(x)=f_N(x)\cdot s_d$ yields the $d$th bit of $f_N(x)$,
we use the corresponding hypotheses $h_{s_d}(x)$ to construct an estimate of $f_N(x)$ as
\begin{equation}
    \tilde{f}_N(x)=\begin{pmatrix}
    h_{s_1}(x)\\
    h_{s_2}(x)\\
    \vdots\\
    h_{s_D}(x)
    \end{pmatrix}\in\mathbb{F}_2^D.
\end{equation}
Thus, the polynomial-time classical algorithms for evaluating the hypotheses would be able to compute each element of this vector and thus approximate $f_N(x)$ well with high probability, which contradicts the fact that $f_N$ is a quantumly advantageous function.
Therefore, our concept class that includes the concepts $c_{s_d}(x)$ for $d=1,\ldots, D$ is not classically efficiently learnable by any classically efficiently evaluatable hypothesis class (see Appendix~\ref{sec:classical hardness} for details).

\appendix

\section*{Appendices}

The appendices of ``Advantage of Quantum Machine Learning from General Computational Advantages'' are organized as follows.
Appendix~\ref{Sec: Setting and definition} gives settings and definitions of a learning model and quantum computational advantage.
Appendix~\ref{sec:advantage_QML} provides our learning task and proof of the advantage of quantum machine learning (QML) in solving our learning task.
Appendix~\ref{sec: Learning advantage without conditions on the hypothetical class} describes a setup for demonstrating this advantage of QML and presents protocols for preparing the classical sample data for the demonstration.

\section{\label{Sec: Setting and definition}Setting and definition}

In this appendix, we present settings and definitions relevant to our work. 
In Appendix~\ref{sec:pac_learning}, we define a model of probably and approximately correct (PAC) learning~\cite{valiant1984theory, kearns1994introduction, kearns1990computational} to be analyzed in this work and also define the advantage of QML\@.
In Appendix~\ref{sec:quantum_advantage}, we define quantum advantages in computing a function, which we will use to show the advantage of QML\@.

\subsection{\label{sec:pac_learning}PAC learning model}

The analysis in our work will be based on a conventional model of learning called the PAC learning model in Ref.~\cite{kearns1994introduction}.
Let $\mathcal{D}_N$ be any unknown target probability distribution supported on an input space $\mathcal{X}_N\subseteq\{0,1\}^N$ of $N$ bits.
In our work, a concept class $\mathcal{C}_N$ over $\mathcal{X}_N$ is a set of functions from the input space to the set of binary labels.
Consequently, a concept $c\in\mathcal{C}_N$ is a function such that $c:\mathcal{X}_N\to\{0,1\}$.
A sample is denoted as $( x,c(x) )$, which is a pair of the input and the output for the concept.
Let $\mathbf{EX}(c,\mathcal{D}_N)$ be a procedure (oracle) that returns a labeled sample $( x,c(x))$ within a unit time upon each call, where $x$ is drawn randomly and independently according to $\mathcal{D}_N$.
Note that the samples $( x,c(x))$ from $\mathbf{EX}(c,\mathcal{D}_N)$ are given in terms of classical bit strings throughout this paper.
In this setting, we can consider $\mathcal{X}=\bigcup_{N\geq1}\mathcal{X}_N$ and $\mathcal{C}=\bigcup_{N\geq1}\mathcal{C}_N$ to define an infinite family of learning problems of increasing input lengths.

In the PAC learning model, a learning algorithm will have access to samples from $\mathbf{EX}(c,\mathcal{D}_N)$ for an unknown target concept $c$, which is chosen from a given concept class $\mathcal{C}_N$.
Using the samples, the learning algorithm will find a hypothesis $h$ from a hypothesis class $\mathcal{H}_N$ so that $h$ should approximate $c$.
We define a measure of approximation error between hypothesis $h$ and unknown concept $c$ as 
\begin{equation}
\label{eq:error}
    \error (h)=\mathbf{Pr}_{x\sim\mathcal{D}_N}[c(x)\neq h(x)].
\end{equation}

The learning algorithm only sees input-output samples of the unknown target concept $c$.
The algorithm may not have to directly deal with the representation of true $c$, i.e., a symbolic encoding of $c$ in terms of a bit string.
However, it still matters which representation the algorithm chooses for its hypothesis $h$ since the learning algorithm needs to output the representation of $h$.
To deal with these representations more formally, consider a representation scheme $\mathcal{R}$ for a concept class $\mathcal{C}_N$, which is a function $\mathcal{R}:\Sigma^\ast\to\mathcal{C}_N$ with $\Sigma=\{0,1\}$ denoting a bit and $\Sigma^\ast\coloneqq\bigcup_{N\geq1}\Sigma^N$ denoting the set of bit strings.
We call any string $\sigma\in\Sigma^\ast$ such that $\mathcal{R}(\sigma)=c$ a representation of $c$.
For $\mathcal{R}$, we here consider the length of the bit string $\sigma\in\Sigma^\ast$ to be the size of each representation $\sigma$, which we write $\size(\sigma)$.
A representation of hypothesis $h$ can be formulated in the same way by replacing $c$ with $h$ and $\mathcal{C}_N$ with $\mathcal{H}_N$.

The learning task is divided into two main parts.
One is to find, using the samples, a representation of the hypothesis $h$ that approximates the target concept $c$ well, and the other is to make a prediction by evaluating $h(x)$ correctly for a new input $x\in\mathcal{X}_N$ and the learned representation of $h$.
First, we define efficient learnability as shown below.
This definition requires that the algorithm find and output the representation of the appropriate hypothesis $h$ for the target concept $c$ within a polynomial time.
Note that this definition implies that the representation of the hypotheses $h$ should also be of at most polynomial length.
Conventionally, the PAC learning model may require that it be learnable for all distributions~\cite{kearns1994introduction}, but we can here observe that learning tasks in practice usually deal with data given from a particular distribution; for example, in the classification of images of dogs and cats, the learning algorithm does not have to work for any distribution over the images, but it suffices to deal with a given distribution supported on the meaningful images such as those of dogs and cats.
Based on this observation, our model requires that it be learnable for a given (unknown) target distribution.
Note that throughout the learning, the learning algorithm does not have to estimate the description of the target distribution itself, but it only suffices to learn the target concept $c$ by the hypothesis $h$.

\begin{definition}[\label{def: efficiently pac learnable}Efficient learnability]
For any problem size $N\in\mathbb{N}$, let $\mathcal{C}_N$ be a concept class and $\mathcal{D}_N$ be a target distribution over an input space $\mathcal{X}_N\subseteq\{0,1\}^N$ of $N$ bits. 
We say that $\mathcal{C}_N$ is quantumly (classically) efficiently learnable under the target distribution $\mathcal{D}_N$ if there exists a hypothesis class $\mathcal{H}_N$ and a quantum (classical randomized) algorithm $\mathcal{A}$ with the following property:
for every concept $c\in\mathcal{C}_N$, and for all $0 < \epsilon,\delta < 1$, if $\mathcal{A}$ is given access to $\mathbf{EX}(c,\mathcal{D}_N)$ in addition to $\epsilon$ and $\delta$, then $\mathcal{A}$ runs in a polynomial time
\begin{align}
    t_{\mathcal{A}}(\mathbf{EX},\epsilon,\delta)=O(\poly(N,1/\epsilon,1/\delta)),
\end{align}
to output a representation of hypothesis $h\in\mathcal{H}_N$ satisfying, with probability at least $1-\delta$,
\begin{equation}
    \error(h)\leq\epsilon,
\end{equation}
where the left-hand side is given by~\eqref{eq:error}.
The probability is taken over the random examples drawn from the calls of $\mathbf{EX}(c,\mathcal{D}_N)$ and the randomness used in the randomized algorithm $\mathcal{A}$.
The number of calls of $\mathbf{EX}(c,\mathcal{D}_N)$ (i.e., the number of samples) and the size of the output representation of the hypothesis are bounded by the runtime.
\end{definition}

As for the evaluation of the learned hypothesis, we give a definition of efficient evaluatability below.
This definition requires that, given an input $x$ and a representation $\sigma_h$ of a hypothesis $h$, the algorithm should evaluate $h(x)$ correctly in a polynomial time for a large fraction of $x$ with high probability.
The representation of a hypothesis used in this definition can be, in general, an arbitrary polynomial-length bit string representing the hypothesis.
In particular, in the case of the previous work~\cite{liu2021rigorous, servedio2004equivalences} on the advantage of QML using Shor's algorithms for solving integer factoring and discrete logarithms~\cite{shor1994algorithms,doi:10.1137/S0097539795293172,shor1999polynomial}, a classical algorithm may be able to prepare samples on its own; by contrast, in our general setting, samples are given from the oracle $\mathbf{EX}$ as in Definition~\ref{def: efficiently pac learnable}, and we do not necessarily require that the evaluation algorithms should be able to simulate $\mathbf{EX}$ to prepare the samples on their own.
(For example, in Appendix~\ref{sec:quantum_data}, we will discuss the preparation of samples by quantum computation rather than classical computation, so the classical algorithm may not be able to prepare the samples on its own.)
In this learning setting, at best, an evaluation algorithm may be able to use the given samples encoded in the polynomial-length representation of the hypothesis as an extra input to the algorithm, which may not be prepared by the algorithm on its own but can be used for the algorithm to compute $h(x)$ more efficiently~\cite{huang2021power,gyurik2023exponential, gyurik2023establishing}.
In addition to this encoding of a polynomial amount of sample data used in the learning, the representation of the hypothesis can even include any polynomial-length bit string to help the evaluation algorithm compute the hypothesis even more efficiently (which may be prepared potentially using an exponential runtime if we do not assume efficient learnability).
In complexity theory, such an extra input (apart from $x$) to make the computation potentially more efficient is known as an advice string~\cite{arora2009computational}, as described in more detail in Appendix~\ref{sec:quantum_advantage}.
Also, in a conventional setting, efficient evaluation in the PAC model may mean that the hypothesis can be evaluated in worst-case polynomial time \cite{kearns1994introduction}.
However, recalling the above observation that practical learning tasks deal with data from a more specific target distribution, we see that it may be too demanding to require that the hypothesis $h$ should be evaluated efficiently for all possible $x\in\mathcal{X}_N$ even in the worst case; rather, it makes more sense to require that we should be able to evaluate $h(x)$ efficiently for $x$ drawn from the target distribution of interest with a sufficiently high probability.
In the complexity theoretical terms, this requirement can be captured by the notion of heuristic polynomial time~\cite{TCS-004,Impagliazzo}; accordingly, we define efficient evaluatability based on heuristic complexity, as shown below.
Since the heuristic hardness implies the worst-case hardness, proving the hardness of efficient evaluation for our learning model immediately leads to the conventional worst-case hardness of efficient evaluation in the learning (see also Appendix~\ref{sec:quantum_advantage} for more discussion on the difference and relation between these hardness results).
\begin{definition}[\label{def: efficiently evaluatable}Efficient evaluatability]
For any problem size $N\in\mathbb{N}$, let $\mathcal{C}_N$ be a concept class and $\mathcal{D}_N$ be a target distribution over an input space $\mathcal{X}_N\subseteq\{0,1\}^N$ of $N$ bits. 
We say that the hypothesis class $\mathcal{H}_N$  is quantumly (classically) efficiently evaluatable under the target distribution $\mathcal{D}_N$ if there exists a quantum (classical randomized) algorithm $\mathcal{A}$ such that for all $N\in\mathbb{N}$, $0 < \epsilon,\delta < 1$, and a $O(\poly(N,1/\epsilon,1/\delta))$-length bit string $\sigma_h$ representing any hypothesis $h\in\mathcal{H}_N$,
$\mathcal{A}$ runs in a polynomial time for all $x\in\mathcal{X}_N$
\begin{align}
    t_{\mathcal{A}}(x,\sigma_h,\epsilon,\delta)=O(\poly(N,1/\epsilon,1/\delta)),
\end{align}
to output $\mathcal{A}(x,\sigma_h,\epsilon,\delta)$ satisfying
\begin{align}
    \pr_{x\sim\mathcal{D}_N}\left[\pr\qty[\mathcal{A}(x,\sigma_h,\epsilon,\delta) = h(x)]\geq 1-\delta \right] \geq 1-\epsilon,
\end{align}
where the inner probability is taken over the randomness of the randomized algorithm $\mathcal{A}$.
\end{definition}

From Definitions~\ref{def: efficiently pac learnable} and~\ref{def: efficiently evaluatable}, a learning task can be divided into four categories $\CC$, $\CQ$, $\QC$, and $\QQ$~\cite{gyurik2023exponential, gyurik2023establishing}, which means that whether the learning task is classically or quantumly efficiently learnable and whether it is classically or quantumly efficiently evaluatable, respectively.
It is known that the categories $\CQ$ and $\QQ$ are equivalent unless the hypothesis class is fixed~\cite{gyurik2023exponential}.
Note that Refs.~\cite{gyurik2023exponential,gyurik2023exponential} defined these categories based on the worst-case complexity, but the categories for our definitions may be different in that  Definitions~\ref{def: efficiently pac learnable} and~\ref{def: efficiently evaluatable} are given based on the heuristic complexity.
In our work, we will study the advantage of QML in the sense of $\CC/\QQ$ separation, following the previous work~\cite{servedio2004equivalences,liu2021rigorous} on the advantage of QML\@; i.e., we will construct a learning task in $\QQ$ but not in $\CC$\@. 

Finally, we remark that the learnability and the evaluatability are different in that Definition~\ref{def: efficiently pac learnable} requires that the hypothesis $h$ should be found in polynomial time with a high probability, and Definition~\ref{def: efficiently evaluatable} requires that $h$ should be evaluated.
Efficient learnability itself does not require that the learned hypothesis should be efficiently used for making a prediction via its evaluation, and efficient evaluatability itself does not require that the representation of the hypothesis should be obtained efficiently in learning from the samples.
However, to achieve the end-to-end acceleration of QML\@, we eventually need both quantumly efficient learnability and quantumly efficient evaluatability simultaneously, which our analysis aims at.
Correspondingly, for the classical hardness of the learning tasks, our interest is to rule out the possibility that the best classical method achieves efficient learnability and efficient evaluatability simultaneously.

\subsection{\label{sec:quantum_advantage}Quantum computational advantage}
In this appendix, we present the computational complexity classes relevant to our analysis of the advantage of QML\@.
Our analysis will use a general class of functions that are efficient to compute by quantum computation but hard by classical computation, based on the complexity classes defined here.

In the following, we will start by presenting the worst-case, average-case, and heuristic computational complexity classes~\cite{TCS-004,Impagliazzo}, whose difference arises from the fraction of the inputs for which the problem can be solved in polynomial time.
Then, we will present the complexity classes for computing functions by classical randomized algorithms and quantum algorithms.
Finally, we will explain advice strings, which are, roughly speaking, bit strings given to the algorithm in addition to the input to help solve the problem efficiently.

We introduce three cases of complexity classes: worst-case, average-case, and heuristic polynomial time.
A complexity class is conventionally defined as a worst-case class; for example, the class \textsf{P} of (worst-case) polynomial time is a family of decision problems such that all the inputs of the problems in the family can be solved within a polynomial time in terms of the length of the input bit strings (even for the worst-case choice of the input)~\cite{arora2009computational}.
Accordingly, the existing analyses of the advantage of QML in previous research were also based on the worst-case complexity classes, requiring that the algorithm should be able to evaluate the hypothesis $h(x)$ for all $x\in\mathcal{X}_N$~\cite{liu2021rigorous, gyurik2023exponential, gyurik2023establishing}.
However, this requirement is too demanding in our setting; after all, the PAC learning model allows for errors depending on a given target distribution.
In the learning as in Definitions~\ref{def: efficiently pac learnable} and~\ref{def: efficiently evaluatable}, it suffices to learn and evaluate $h(x)$ efficiently only for a sufficiently large fraction of $x$ on the support of the given target distribution, rather than all $x$.
Intuitively, in the classification of images for example, it suffices to learn and evaluate $h(x)$ efficiently for meaningful images on the support of the true probability distribution of samples, and whether $h(x)$ can be evaluated for all possible images including those never appearing in the real-world data is irrelevant in practice.
To capture this difference, we here take into account average-case and heuristic complexity classes~\cite{Impagliazzo, TCS-004}.
The class of worst-case polynomial time is the class of decision problems that are solvable in polynomial time.
Whereas \textsf{P} is a class of decision problems, average-case polynomial time \textsf{AvgP} and heuristic polynomial time \textsf{HeurP} are the classes of \textit{distributional} decision problems that consist not only of decision problems but also of the target probability distributions of the input.
In solving the distributional problems, the runtime of an algorithm may probabilistically change depending on the input given from the distribution.
The runtime of solving the problem in \textsf{AvgP} should be polynomial in input length on average over inputs given from the distribution~\cite{levin1986average}, which may allow, e.g., an exponentially long runtime to obtain a correct answer for an exponentially small fraction of the inputs.
On the other hand, \textsf{HeurP} requires that the runtime of correctly solving the problem should be polynomial only for a sufficiently large fraction of (yet not all) the inputs from the distribution, and for the rest of the small fraction of the inputs, the algorithm may output a wrong answer~\cite{Impagliazzo, TCS-004}.
Note that in the heuristic class, the average runtime of correctly solving the problem is not necessarily bounded.
By definition, the worst-case complexity class is contained in the average-case class, and the average-case class in the heuristic class; i.e., it is shown that $\P \subseteq \textsf{AvgP} \subseteq \textsf{HeurP}$~\cite{1443089}.
More formally, we define the worst-case, average-case, and heuristic classes with reference to \textsf{P} as follows.
Note that for our analysis of learning, only the worst-case and heuristic classes are relevant, while the average-case classes are discussed here for clarity of presentation; thus, we may stop mentioning the average-case classes after this definition.
\begin{definition}[\label{def: class worst- average- heur-case based on P}Worst-case, average-case, and heuristic polynomial time~\cite{levin1986average, Impagliazzo, TCS-004}]
We define the worst-case, average-case, and heuristic complexity classes, namely, $\P$, $\textsf{AvgP}$, and $\textsf{HeurP}$, respectively, as follows.
\begin{enumerate}
    \item A decision problem, i.e., a function $L:\{0,1\}^\ast\to\{0,1\}$ with a single-bit output, is in $\P$ if there exists a deterministic classical algorithm $\mathcal{A}$ such that for every $N$ and every input $x \in \{0,1\}^N$, $\mathcal{A}$ outputs $L(x)$ in $\poly(N)$ time.
    \item A distributional problem $(L,\mathcal{D})$ is in $\textsf{AvgP}$ if there exists a deterministic classical algorithm $\mathcal{A}$ and a constant $d$ such that for every $N$ 
    \begin{align}
        \mathbb{E}_{x\sim\mathcal{D}_N}\left[\frac{{t_\mathcal{A}(x)}^{\frac{1}{d}}}{N}\right]=O(1),
    \end{align}
    where $t_{\mathcal{A}}(x)$ is the time taken to calculate $L(x)$ by $\mathcal{A}$, and $\mathbb{E}_{x\sim\mathcal{D}_N}[\cdots]$ is the expected value over $x$ drawn from $\mathcal{D}_N$.
    \item A distributional problem $(L,\mathcal{D})$ is in $\textsf{HeurP}$ if there exists a deterministic classical algorithm $\mathcal{A}$ such that for every $N$ and all $0< \mu < 1$, the runtime $t_{\mathcal{A}}(x,\mu)$ of $\mathcal{A}$ for every input $x$ in the support of $\mathcal{D}_N$ is $t_{\mathcal{A}}(x,\mu)=O(\poly(N,1/\mu))$, and the output $\mathcal{A}(x,\mu)$ of $\mathcal{A}$ satisfies
    \begin{align}
        \pr_{x\sim\mathcal{D}_N}[\mathcal{A}(x,\mu)=L(x)]\geq 1-\mu.
    \end{align}
\end{enumerate}
\end{definition}

Next, we define the classes of problems solvable by a (classical) randomized algorithm.
In Definition~\ref{def: class worst- average- heur-case based on P}, we use a deterministic classical algorithm to solve problems.
By contrast, in the randomized algorithms, we need to take into account errors arising from the randomization.
Corresponding to \textsf{P} and \textsf{HeurP}, we define the two classes of problems for randomized algorithms as follows.
\begin{definition}[Worst- and heuristic bounded-error probabilistic polynomial time~\cite{Impagliazzo, TCS-004}]
We define a worst-case class $\BPP$ and a heuristic class $\textsf{HeurBPP}$ for classical randomized algorithms as follows.
\begin{enumerate}
    \item A decision problem $L$ is in $\textsf{BPP}$ if there exists a classical randomized algorithm $\mathcal{A}$ such that for every $N$ and every input $x \in \{0,1\}^N$,
    the runtime $t_{\mathcal{A}}(x)$ of $\mathcal{A}$ is $t_{\mathcal{A}}(x)=O(\poly(N))$, and the output $\mathcal{A}(x)$ of $\mathcal{A}$ satisfies
    \begin{align}
        \pr[\mathcal{A}(x) = L(x)] \geq 2/3,
    \end{align}
    where the probability is taken over the randomness of $\mathcal{A}$.
    \item A distributional problem $(L,\mathcal{D})$ is in $\textsf{HeurBPP}$ if there exists a classical randomized algorithm such that for every $N$ and all $0< \mu <1$, the runtime $t_{\mathcal{A}}(x,\mu)$ of $\mathcal{A}$ for every input $x$ in the support of $\mathcal{D}_N$ is $t_{\mathcal{A}}(x,\mu)=O(\poly(N,1/\mu))$, and the output $\mathcal{A}(x,\mu)$ of $\mathcal{A}$ satisfies
    \begin{align}
        \pr_{x\sim\mathcal{D}_N}[\pr[\mathcal{A}(x,\mu)=L(x)]\geq 2/3]\geq 1-\mu,
    \end{align}
    where the inner probability of $\mathcal{A}(x,\mu)=L(x)$ is taken over randomness of $\mathcal{A}$.
\end{enumerate}
\end{definition}

Whereas we have so far explained complexity classes of decision problems, i.e., those for computing functions with a single-bit output, our analysis will use a Boolean function with a single multi-bit output for each $N$-bit input
\begin{equation}
    f_N:\{0,1\}^N\to\{0,1\}^{D(N)},
\end{equation}
where $D:\mathbb{N}\to\mathbb{N}$ is any function satisfying $D(N)=O(\poly(N))$, and we may abbreviate $D(N)$ as $D$ in the following of this paper.
Accordingly, we define the complexity classes of function problems, i.e., problems of computing such multi-bit output functions $f_N$.
For the heuristic complexity class, we also refer to a family of problems $\{(f_N,\mathcal{D}_N)\}_{N\in\mathbb{N}}$ as distributional function problems.
Whenever $f_N$ is used in the following of this paper, it refers to a function with a single multi-bit output for each input.
Note that the complexity classes of function problems may also be defined as those of search problems, which can be considered to be the problems of computing functions with many possible outputs for each input, and the algorithms aim to search for one of the possible outputs for a given input.
But even if one considers such a more general definition, functions $f_N$ relevant to our analysis are those with a single output for each input; correspondingly, we here present the definitions using the single-output functions for simplicity. 

\begin{definition}[Worst-case and heuristic distributional function bounded-error polynomial time]
We define a worst-case class $\textsf{FBPP}$ and a heuristic class $\textsf{HeurFBPP}$ for computing multi-bit output functions as follows.
\begin{enumerate}
    \item 
    Given $R_N \coloneqq \{(x,f_N(x))\}_{x\in\{0,1\}^N}$ and $R \coloneqq \bigcup_N R_N$, the relation $R$ is in $\textsf{FBPP}$ if there exists a randomized classical algorithm $\mathcal{A}$ such that for all $N$, every input $x \in \{0,1\}^N$, and all $0 < \nu < 1$,
    the runtime $t_{\mathcal{A}}(x,\nu)$ of $\mathcal{A}$ is $t_{\mathcal{A}}(x,\nu)=O(\poly(N,1/\nu))$, and the output $\mathcal{A}(x,\nu)$ of $\mathcal{A}$ satisfies
    \begin{align}
        \pr[(x,\mathcal{A}(x,\nu)) \in R_N ] \geq 1-\nu,
    \end{align}
    where the probability is taken over the randomness of $\mathcal{A}$. 
    \item A distributional function problem $F = \{(f_N,\mathcal{D}_N)\}_{N\in\mathbb{N}}$ is in $\textsf{HeurFBPP}$ if there exists a classical randomized algorithm $\mathcal{A}$ such that for all $N$ and all $0 < \mu, \nu < 1$,
    the runtime $t_{\mathcal{A}}(x,\mu,\nu)$ of $\mathcal{A}$ for every input $x \in \{0,1\}^N$ in the support of $\mathcal{D}_N$ is $t_{\mathcal{A}}(x,\mu,\nu)=O(\poly(N,1/\mu,1/\nu))$, and the output $\mathcal{A}(x,\mu,\nu)$ of $\mathcal{A}$ satisfies
    \begin{align}
        \pr_{x\sim\mathcal{D}_N}[\pr[\mathcal{A}(x,\mu,\nu)=f_N(x)]\geq 1-\nu]\geq 1-\mu,
    \end{align}
    where the inner probability of $\mathcal{A}(x,\mu,\nu)=f_N(x)$ is taken over randomness of $\mathcal{A}$.
\end{enumerate}
\end{definition}

We next define the classes $\textsf{FBQP}$ and $\textsf{HeurFBQP}$ of problems efficiently solvable by quantum algorithms.
The classes defined so far are the computational complexity classes for deterministic or randomized classical algorithms, but we here define $\textsf{FBQP}$ and $\textsf{HeurFBQP}$ using quantum algorithms in place of the classical algorithms.
The class $\textsf{HeurFBQP}$ will be used for our construction of learning tasks in Definition~\ref{def:quantum_advantage} of Appendix~\ref{sec:formulation of learning task} to prove the advantage of QML\@.
Note that we have $\textsf{FBQP}\subseteq\textsf{HeurFBQP}$ in the same way as $\textsf{P}\subseteq\textsf{HeurP}$.
Working on $\textsf{HeurFBQP}$, we aim to make it possible to use a potentially larger class of computational advantages of heuristic quantum algorithms captured by $\textsf{HeurFBQP}$ rather than $\textsf{FBQP}$, so as to achieve a wider class of learning tasks more efficiently.
\begin{definition}[Worst-case and heuristic distributional function bounded-error quantum polynomial time]
We define a worst-case class $\textsf{FBQP}$ and a heuristic class $\textsf{HeurFBQP}$ for quantum algorithms as follows.
\begin{enumerate}
\item Given $R_N = \{(x,f_N(x))\}_{x\in\{0,1\}^N}$ and $R = \bigcup_N R_N$,
    the relation $R$ is in $\textsf{FBQP}$ if there exists a quantum algorithm $\mathcal{A}$ such that for all $N$, every input $x \in \{0,1\}^N$, all $0 < \nu < 1$,
    the runtime $t_{\mathcal{A}}(x,\nu)$ of $\mathcal{A}$ is $t_{\mathcal{A}}(x,\nu)=O(\poly(N,1/\nu))$, and the output $\mathcal{A}(x,\nu)$ of $\mathcal{A}$ satisfies
    \begin{align}
        \pr[(x,\mathcal{A}(x,\nu)) \in R_N ] \geq 1-\nu,
    \end{align}
    where the probability is taken over the randomness of $\mathcal{A}$. 
    \item A distributional function problem $F = \{(f_N,\mathcal{D}_N)\}_{N\in\mathbb{N}}$ is in $\textsf{HeurFBQP}$ if there exists a quantum algorithm $\mathcal{A}$ such that for all $N$ and all $0 < \mu, \nu < 1$,
     the runtime $t_{\mathcal{A}}(x,\mu,\nu)$ of $\mathcal{A}$ for every input $x \in \{0,1\}^N$ in the support of $\mathcal{D}_N$ is $t_{\mathcal{A}}(x,\mu,\nu)=O(\poly(N,1/\mu,1/\nu))$, and the output $\mathcal{A}(x,\mu,\nu)$ of $\mathcal{A}$ satisfies
    \begin{align}
    \label{eq:heurfbqp}
        \pr_{x\sim\mathcal{D}_N}[\pr[\mathcal{A}(x,\mu,\nu)=f_N(x)]\geq 1-\nu]\geq 1-\mu,
    \end{align}
    where the inner probability of $\mathcal{A}(x,\mu,\nu)=f_N(x)$ is taken over randomness of $\mathcal{A}$.
\end{enumerate}
\end{definition}

Finally, we introduce the complexity classes with advice strings.
As discussed in Appendix~\ref{sec:pac_learning}, the analysis of the efficient evaluatability in the PAC learning model needs to take into account the advice strings of at most $O(\poly(N,1/\epsilon,1/\delta))$ length~\cite{huang2021power,gyurik2023exponential, gyurik2023establishing}.
To capture the power of the bit strings representing the hypotheses to be evaluated, we consider the complexity classes with a polynomial-length advice string as follows.

\begin{definition}[\label{def:HeurFBPP/poly}Worst-case and heuristic distributional function bounded-error polynomial time with advice]
We define a worst-case class $\textsf{FBPP/poly}$ and a heuristic class $\textsf{HeurFBPP/poly}$ with advice as follows.
\begin{enumerate}
    \item Given $R_N = \{(x,f_N(x))\}_{x\in\{0,1\}^N}$ and $R = \bigcup_N R_N$,
    the relation $R$ is in $\textsf{FBPP/poly}$ if there exists a randomized classical algorithm $\mathcal{A}$ such that for all $N$, every input $x \in \{0,1\}^N$, and all $0 < \nu < 1$, there exists an advice string $\alpha_{N,\nu} \in \{0,1\}^{O(\poly(N,1/\nu))}$ such that
    the runtime $t_{\mathcal{A}}(x,\alpha_{N,\nu},\nu)$ of $\mathcal{A}$ is $t_{\mathcal{A}}(x,\alpha_{N,\nu},\nu)=O(\poly(N,1/\nu))$, and the output $\mathcal{A}(x,\alpha_{N,\nu},\nu)$ of $\mathcal{A}$ satisfies
    \begin{align}
        \pr[(x,\mathcal{A}(x,\alpha_{N,\nu},\nu)) \in R_N ] \geq 1-\nu,
    \end{align}
    where the probability is taken over the randomness of $\mathcal{A}$. 
    \item A distributional function problem $F = \{(f_N,\mathcal{D}_N)\}_{N\in\mathbb{N}}$ is in $\textsf{HeurFBPP/poly}$ if there exists a randomized classical algorithm $\mathcal{A}$ such that for all $N$ and all $0 < \mu, \nu < 1$, there exists an advice string $\alpha_{N,\mu,\nu} \in \{0,1\}^{O(\poly(N,1/\mu, 1/\nu))}$ such that
    the runtime $t_{\mathcal{A}}(x,\alpha_{N,\mu,\nu},\mu,\nu)$ of $\mathcal{A}$ for every input $x \in \{0,1\}^N$ in the support of $\mathcal{D}_N$ is $t_{\mathcal{A}}(x,\alpha_{N,\mu,\nu},\mu,\nu)=O(\poly(N,1/\mu,1/\nu))$, and the output $\mathcal{A}(x,\alpha_{N,\mu,\nu},\mu,\nu)$ of $\mathcal{A}$ satisfies
    \begin{align}
    \label{eq:heurfbpp_poly}
        \pr_{x\sim\mathcal{D}_N}[\pr[\mathcal{A}(x,\alpha_{N,\mu,\nu},\mu,\nu)=f_N(x)]\geq 1-\nu]\geq 1-\mu,
    \end{align}
    where the probability of $\mathcal{A}(x,\alpha_{N,\mu,\nu},\mu,\nu)=f_N(x)$ is taken over randomness of $\mathcal{A}$.
\end{enumerate}
\end{definition}
We similarly define other possible classes such as $\textsf{FP}$ by combining the above definitions.

\section{\label{sec:advantage_QML}Advantage of QML from general quantum computational advantages}

In this appendix, we prove that, for general quantum computational advantages, we can correspondingly construct learning tasks that are hard for classical computation but can be efficiently solved by quantum computation, within the conventional framework of supervised learning (i.e., in the PAC model formulated in Appendix~\ref{sec:pac_learning}).
In previous work~\cite{liu2021rigorous, servedio2004equivalences}, the advantage of QML was observed only under the computational hardness assumption for a specific type of problem, such as that solved by Shor's algorithms.
In contrast, the learning tasks introduced here will be based on general types of quantum computational advantages, i.e., arbitrary functions in $\textsf{HeurFBQP}\setminus(\textsf{HeurFBPP}/\textsf{poly})$ rather than just that of Shor's algorithms.
In Appendix~\ref{sec:formulation of learning task}, we explicitly give the concept class of these learning tasks as linear separation problems in the space of bits.
In Appendix~\ref{sec:quantum_algorithm}, we construct polynomial-time quantum algorithms for learning and evaluation in our learning tasks.
In Appendix~\ref{sec:classical hardness}, we rigorously prove the classical hardness of the learning tasks.

\subsection{\label{sec:formulation of learning task}Formulation of learning tasks}

In this appendix, we construct learning tasks using general types of quantum advantages based on complexity classes introduced in Appendix~\ref{sec:quantum_advantage}.

First, we define the general complexity class of functions to be used for formulating our learning task.
Although Refs.~\cite{gyurik2023establishing, gyurik2023exponential} studied conditions on the complexity classes that potentially lead to the advantage of QML, the analyses in Refs.~\cite{gyurik2023establishing, gyurik2023exponential} were based on worst-case complexity~\cite{gyurik2023establishing,gyurik2023exponential} and were not able to explicitly construct the learning tasks satisfying their conditions in general.
By contrast, we here identify an appropriate class of functions using the heuristic complexity classes, so that we can use any functions in this class for our explicit construction of the learning tasks with the provable advantage of QML\@.
The class that we use is given as follows.

\begin{definition}[\label{def:quantum_advantage}Quantumly advantageous functions]
    For a distributional functions problem $\{(f_N,\mathcal{D}_N)\}_{N\in\mathbb{N}}$ in
    \begin{equation}
        \{(f_N,\mathcal{D}_N)\}_N \in \HeurQ\setminus(\HeurC),
    \end{equation}
    we call $f_N$ a quantumly advantageous function under the target distribution $\mathcal{D}_N$. 
\end{definition}

As presented in the main text, the quantumly advantageous functions may include various functions beyond those computed by Shor's algorithms.
Since we use heuristic complexity classes, the class of functions in Definition~\ref{def:quantum_advantage} is even larger than the class defined by the worst-case complexity classes, as discussed in Appendix~\ref{sec:quantum_advantage}.
For example, our definition does not rule out the possibility of using heuristic quantum algorithms such as the variational quantum algorithms (VQAs) for seeking evidence of the utility of QML based on the heuristic complexity class~\cite{cerezo2021variational}, in case one finds a variant of such quantum algorithms that are faster than classical algorithms for most of the inputs.

We define our concept class using the quantumly advantageous functions below.
In the existing work \cite{liu2021rigorous, servedio2004equivalences} on the advantage of QML, the target distribution was limited to the uniform distribution, and the task was dependent on the specific mathematical structure of the functions computed by Shor's algorithms; by contrast, we allow for an arbitrary target distribution $\mathcal{D}_N$ over $N$ bits, and we can use an arbitrary quantumly advantageous function without specifically depending on Shor's algorithms.

\begin{definition}[\label{def: concept class based on quantum advantage}Concept class for the advantage of QML from general computational advantages]
For any $N$, $D = O(\poly(N))$, any target distribution $\mathcal{D}_N$ over an input space $\mathcal{X}_N\subseteq\{0,1\}^N$ of $N$ bits, and any quantumly advantageous function $f_N:\{0,1\}^N\to\mathbb{F}_2^{D}$ under $\mathcal{D}_N$,
we define a concept class $\mathcal{C}_N$ over the input space $\mathcal{X}_N$ as $\mathcal{C}_N=\{c_s\}_{s\in{\mathbb{F}_2^D}}$ with its concept $c_s$ for each parameter $s\in\mathbb{F}_2^D$ given by
\begin{equation}
\label{eq:concept_class}
c_s(x)\coloneqq f_N(x)\cdot s\in\mathbb{F}_2=\{0,1\},
\end{equation}
where $f_N(x)\cdot s$ for $f_N(x),s\in\mathbb{F}_2^D=\{0,1\}^D$ is a bitwise inner product in the vector space $\mathbb{F}_2^D$ over the finite field. 
\end{definition}

\subsection{\label{sec:quantum_algorithm}Construction of polynomial-time quantum algorithms for learning and evaluation}

In this appendix, we show polynomial-time quantum algorithms for learning concepts in the concept class in Definition~\ref{def: concept class based on quantum advantage} and for evaluating hypotheses in the hypothesis class for this concept class.
We first describe our learning algorithm (Algorithm~\ref{alg:algorithm of learning}) and prove the quantum efficient learnability for our concept class.
We then describe our evaluation algorithm (Algorithm~\ref{alg:algorithm of evaluation}) and prove the quantum efficient evaluatability for the hypothesis class constructed for our concept class.
Note that the proof of the classical hardness of this learning task for any classical algorithm will also be given in Appendix~\ref{sec:classical hardness}.

We first describe our quantum algorithm for learning.
Our algorithm for learning a target concept in our concept class is given by Algorithm~\ref{alg:algorithm of learning}.
The concept class $\mathcal{C}_N=\{c_s:s\in\mathbb{F}_2^D\}$ is defined in Definition~\ref{def: concept class based on quantum advantage} for any (unknown) target distribution $\mathcal{D}_N$ over $\mathcal{X}_N\subseteq\mathbb{F}_2^N$ and any quantumly advantageous function $f_N:\{0,1\}^N\to\mathbb{F}_2^D$.
Let $c_s$ denote the unknown target concept to be learned from the samples by the algorithm, where $s$ is the true parameter of the target concept.
For any $\epsilon>0$ and $\delta>0$, our algorithm aims to achieve the learning in Definition~\ref{def: efficiently pac learnable}, i.e., to output $\tilde{s}$ so that a hypothesis $h_{\tilde{s}}$ represented by $\tilde{s}$ should satisfy
\begin{equation}
\label{eq:learning_quantum_algorithm}
\pr_{x\sim\mathcal{D}_N}\qty[h_{\tilde{s}}(x)\neq c_s(x)]\leq\epsilon
\end{equation}
with a high probability greater than or equal to $1-\delta$.
To this goal, we set the internal parameters in Algorithm~\ref{alg:algorithm of learning} as
\begin{align}
\label{eq:N_quantum_algorithm}
    M &= \left\lceil\frac{D}{\epsilon}-1\right\rceil,\\
\label{eq:mu_quantum_algorithm}
    \mu &= \frac{\delta}{2M},\\
\label{eq:nu_quantum_algorithm}
    \nu &= \frac{\delta}{2M},
\end{align}
where $\lceil x\rceil$ is the ceiling function, i.e., the smallest integer greater than or equal to $x$.

In Algorithm~\ref{alg:algorithm of learning}, $M$ samples are initially loaded as the input, obtained from the oracle $\mathbf{EX}$ in the setting of the PAC learning model described in Appendix~\ref{sec:pac_learning}.
The $M$ samples are denoted by $\{( x_m,c_s(x_m) )\}_{m=1}^M$, where $c_s$ is the (unknown) target concept to be learned from the samples by the algorithm.
Then, the algorithm probabilistically computes the quantumly advantageous function $f_N$ for each of the $M$ input samples $x_1,\ldots,x_M$.
By definition of the quantumly advantageous function $f_N$ in $\HeurQ$ of~\eqref{eq:heurfbqp},
we have a quantum algorithm $\mathcal{A}$ to achieve
\begin{equation}
\label{eq:quantum_algorithm_function}
    \pr_{x\sim\mathcal{D}_N}\qty[\pr\qty[\mathcal{A}(x,\mu,\nu)=f_N(x)]\geq 1-\nu]\geq 1-\mu,
\end{equation}
with runtime
\begin{equation}
\label{eq:runtime_t_A}
    t_{\mathcal{A}}(x,\mu,\nu)=O\qty(\qty(\frac{N}{\mu\nu})^\alpha),
\end{equation}
where $\alpha>0$ is an upper bound of the degree of the polynomial runtime.
To compute $f_N$, Algorithm~\ref{alg:algorithm of learning} applies the quantum algorithm $\mathcal{A}$ to each of $x_1,\ldots,x_M$.
We write the outputs of $\mathcal{A}$ as $\mathcal{A}(x_1),\ldots,\mathcal{A}(x_M)\in\mathbb{F}_2^D$, respectively, where we will omit $\mu$ and $\nu$ for simplicity of notation if it is obvious from the context.
Note that $\mathcal{A}$ may not be a deterministic algorithm, and thus, we may have $\mathcal{A}(x_m)=f_N(x_m)$ only probabilistically.
Using $\mathcal{A}(x_1),\ldots,\mathcal{A}(x_M)$ obtained from these computations, the algorithm performs Gaussian elimination by classical computation to solve a system of linear equations
\begin{equation}
\label{eq:linear_equation}
\begin{aligned}
\mathcal{A}(x_1)\cdot \tilde{s}&=c_s(x_1),\\
\mathcal{A}(x_2)\cdot \tilde{s}&=c_s(x_2),\\
&\vdots\\
\mathcal{A}(x_M)\cdot \tilde{s}&=c_s(x_M),
\end{aligned}
\end{equation}
where the left-hand sides of the system of linear equations are the bitwise inner product in the space $\mathbb{F}_2^D$ of the $D$-dimensional vectors over the finite field.
This step provides a solution
\begin{equation}
\label{eq:tilde_s}
    \tilde{s}=\begin{pmatrix}
    \tilde{s}_1\\
    \tilde{s}_2\\
    \vdots\\
    \tilde{s}_{D}
    \end{pmatrix}\in\mathbb{F}_2^D
\end{equation}
of the system of linear equations.
This system of linear equations always has the true parameter $s$ of the target concept $c_s$ as a solution but may have more than one solution if the set $\{\mathcal{A}(x_1),\ldots,\mathcal{A}(x_M)\}$ does not include a spanning set of $D$ vectors in the $D$-dimensional vector space $\mathbb{F}_2^D$.
The non-spanning cases indeed occur in our setting, especially when the support of $\mathcal{D}_N$ or the range of $f_N$ is small, on which we impose no assumption for the generality of our learning task.
Even if the system of linear equations has more than one solution, the algorithm can nevertheless adopt any solution of~\eqref{eq:linear_equation} as $\tilde{s}$ in~\eqref{eq:tilde_s}.
The learning algorithm outputs this parameter $\tilde{s}$ as a representation of the hypothesis given by
\begin{equation}
\label{eq:hypothesis}
    h_{\tilde{s}}(x)=f_N(x)\cdot \tilde{s},
\end{equation}
where the right-hand side is the bitwise inner product in $\mathbb{F}_2^D$.
The hypothesis class is then given by
\begin{equation}
\label{eq:hypothesis_class}
    \mathcal{H}_N\coloneqq\{h_{\tilde{s}}:\tilde{s}\in\mathbb{F}_2^D\}.
\end{equation}

\begin{figure}[!t]
    \centering
    \begin{algorithm}[H]
        \caption{Quantum algorithm for learning a concept in the concept class in Definition~\ref{def: concept class based on quantum advantage}}
        \label{alg:algorithm of learning}
        \begin{algorithmic}[1]
            \REQUIRE Samples loaded from the oracle $\mathbf{EX}$, $\epsilon>0$, and $\delta>0$.
            \ENSURE A $D$-bit representation $\tilde{s}\in\mathbb{F}_2^D$ of the hypothesis $h_{\tilde{s}}$ in~\eqref{eq:hypothesis} in the hypothesis class in~\eqref{eq:hypothesis_class} achieving the error below $\epsilon$ with high probability at least $1-\delta$, as in~\eqref{eq:learning_quantum_algorithm}.
            \STATE Load $M$ samples $( x_1, c_s(x_1) ),\ldots,( x_M, c_s(x_M) )$ from the oracle $\mathbf{EX}$ with $M$ given in~\eqref{eq:N_quantum_algorithm}.
            \FOR{$m=1,\dots,M$}
            \STATE Perform the quantum algorithm $\mathcal{A}$ in~\eqref{eq:quantum_algorithm_function} for the input $x_m$ with the parameters $\mu$ and $\nu$ in~\eqref{eq:mu_quantum_algorithm} and~\eqref{eq:nu_quantum_algorithm}, respectively, to obtain $\mathcal{A}(x_m)$.
            \ENDFOR
            \STATE Perform Gaussian elimination by classical computation for solving the system of linear equations in~\eqref{eq:linear_equation}, using $\mathcal{A}(x_1),\ldots,\mathcal{A}(x_M)$ obtained in the previous steps and the output samples $c_s(x_1),\ldots,c_s(x_M)$ loaded initially, to obtain a solution $\tilde{s}$ in~\eqref{eq:tilde_s}.
            \RETURN $\tilde{s}$.
        \end{algorithmic}
    \end{algorithm}
\end{figure}

\begin{figure}[!t]
    \centering
    \begin{algorithm}[H]
        \caption{Quantum algorithm for evaluating a hypothesis in the hypothesis class for the concept class in Definition~\ref{def: concept class based on quantum advantage}}
        \label{alg:algorithm of evaluation}
        \begin{algorithmic}[1]
            \REQUIRE A new input $x\in\mathcal{X}_N$ sampled from the target distribution $\mathcal{D}_N$, a parameter $\tilde{s}\in\mathbb{F}_2^D$ of the hypothesis $h_{\tilde{s}}$ in~\eqref{eq:hypothesis} in the hypothesis class~\eqref{eq:hypothesis_class}, $\epsilon>0$, and $\delta>0$.
            \ENSURE An estimate $\tilde{h}\in\{0,1\}$ of the hypothesis $h_{\tilde{s}}(x)$ for the input $x$ achieving the error below $\epsilon$ with high probability at least $1-\delta$, as in~\eqref{eq:evaluation error}.
            \STATE Perform the quantum algorithm $\mathcal{A}$ in~\eqref{eq:quantum_algorithm_function} for the input $x$ with the parameters $\mu$ and $\nu$ in~\eqref{eq:mu = varep} and~\eqref{eq:nu = del}, respectively, to obtain $\mathcal{A}(x)$.
            \RETURN $\tilde{h}=A_N(x)\cdot \tilde{s}$ in~\eqref{eq:output of evaluation Algorithm}.
        \end{algorithmic}
    \end{algorithm}
\end{figure}

In the following, we will prove the efficient learnability of our concept class $\mathcal{C}_N$ by Algorithm~\ref{alg:algorithm of learning}.
The proof is nontrivial since the parameter $\tilde{s}$ in~\eqref{eq:tilde_s} output by our learning algorithm may not be exactly equal to true $s$ of the target concept $c_s$ but can be any of multiple possible solutions of the system of linear equations in~\eqref{eq:linear_equation}; i.e., we need to take into account the cases of 
\begin{equation}
    \tilde{s}\neq s.
\end{equation}
We will nevertheless prove that we have
\begin{equation}
    h_{\tilde{s}}(x)=c_s(x)
\end{equation}
for a large fraction of $x$ with a high probability as required for the efficient learnability in Definition~\ref{def: efficiently pac learnable}.

To achieve this proof, our key technique is to use the lemma below, which indicates that if we have sufficiently many samples $x_1,\ldots,x_M$, then for a new $(M+1)$th input $x_{M+1}$ to be given in the future, we will be able to represent its feature $y_{M+1}=f_N(x)\in\mathbb{F}_2^D$ as a linear combination of those of the $M$ samples, $y_1=f_N(x_1),\ldots,y_M=f_N(x_M)\in\mathbb{F}_2^D$, with a high probability.
Using this lemma, in our proof of efficient learnability, we will show that the learned hypothesis
$h_{\tilde{s}}(x)=f_N(x)\cdot\tilde{s}$ with $\tilde{s}$ estimated from $y_1,\ldots,y_M$ will coincide with the target concept $c_s(x)=f_N(x)\cdot s$ with true $s$, by expanding $f_N(x)$ therein as the linear combination of $f_N(x_1),\ldots,f_N(x_M)$.
In particular, we here give the following lemma.

\begin{lemma}[\label{lemma:linear combination bound}Probability of linear combination]
    Suppose that $M$ vectors $y_1,\ldots,y_{M}\in\mathbb{F}_2^D$ are sampled from any probability distribution on a $D$-dimensional vector space $\mathbb{F}_2^D$ over the finite field in an identically and identically distributed (IID) way. If the $(M+1)$th vector $y$ is sampled from the same distribution, then $y$ can be represented by a linear combination of the other $M$ vectors $y_1,\ldots,y_{M}$, i.e.,
    \begin{equation}
    \label{eq:probability_linear_combination}
        y=\sum_{m=1}^M \alpha_m y_m\quad\text{for some $\alpha_m\in\mathbb{F}_2=\{0,1\}$},
    \end{equation}
    with a high probability greater than or equal to
    \begin{equation}
        1-\frac{D}{M+1}.
    \end{equation}
\end{lemma}

\begin{proof}
We write $y_{M+1}\coloneqq y$.
Given any sequence $y_1,\ldots,y_{M+1}$ of the $M+1$ vectors, let $m^\prime$ be the number of nonzero vectors in $(y_1,\ldots,y_{M+1})$ such that the vector cannot be represented by a linear combination of the other $M$ vectors.
Let $p(m^\prime)$ denote the probability that the sequence $y_1,\ldots,y_{M+1}$ randomly chosen by the IID sampling includes exactly $m^\prime$ vectors that cannot be represented by a linear combination of the other $M$.
Since the space $\mathbb{F}_2^D$ is $D$-dimensional, we always have
\begin{equation}
\label{eq:m_prime}
    m^\prime\leq D,
\end{equation}
that is,
\begin{equation}
    \sum_{m^\prime=0}^{D}p(m^\prime)=1.
\end{equation}
For example, we may have $m^\prime=D$ in the cases where the sequence includes the $D$ vectors that form a basis of the vector space $\mathbb{F}_2^D$, and the other $N-D$ vectors are zero vectors.

Conditioned on having these $m^\prime$ vectors in the sequence of $M+1$ vectors, the probability of~\eqref{eq:probability_linear_combination} is bounded by the probability of having one of the $m^\prime$ vectors out of the $M+1$ vectors as the $(M+1)$th vector, i.e.,
\begin{align}
    &\pr\left[y_{M+1}\neq \sum_{m=1}^M \alpha_m y_m\quad\text{$\forall \alpha_m\in\mathbb{F}_2$}\middle|m^\prime\right]\nonumber\\
    &=\frac{m^\prime}{M+1}\\
    &\leq\frac{D}{M+1},
\end{align}
where the first equality follows from the assumption of IID sampling, and the inequality in the last line from~\eqref{eq:m_prime}.
Therefore, it holds that
\begin{align}
    &\pr\left[y_{M+1}\neq \sum_{m=1}^M \alpha_m y_m\quad\text{$\forall \alpha_m\in\mathbb{F}_2$}\right]\nonumber\\
    &=\sum_{m^\prime=0}^{M}p(m^\prime)\pr\left[y_{M+1}\neq \sum_{m=1}^M \alpha_m y_m\quad\text{$\forall \alpha_m\in\{0,1\}$}\middle|m^\prime\right]\nonumber\\
    &\leq\left(\sum_{m^\prime=0}^{M}p(m^\prime)\right)\frac{D}{M+1}\\
    &=\frac{D}{M+1},
\end{align}
which yields the conclusion.
\end{proof}

Using Lemma~\ref{lemma:linear combination bound}, we prove that the concept class $\mathcal{C}_N$ in Definition~\ref{def: concept class based on quantum advantage} is quantumly efficiently learnable as follows.

\begin{theorem}[\label{thm: quantumly learnable}Quantumly efficient learnability]
    For any $N$, $D=O(\poly(N))$, any target distribution $\mathcal{D}_N$ over the $N$-bit input space $\mathcal{X}_N\subseteq\{0,1\}^N$, and any quantumly advantageous function $f_N:\{0,1\}^N\to\mathbb{F}_2^D$ under $\mathcal{D}_N$, the concept class $\mathcal{C}_N$ in Definition~\ref{def: concept class based on quantum advantage} with $f_N$ is quantumly efficiently learnable by Algorithm~\ref{alg:algorithm of learning}.
\end{theorem}

\begin{proof}
    In the following, we will first discuss the success probability of our algorithm and then analyze the error in the learning.
    Finally, we will provide an upper bound of the runtime.
    
    Regarding the success probability of Algorithm~\ref{alg:algorithm of learning}, the probabilistic parts of the learning algorithm are the loading of the $M$ samples $( x_1,c_s(x_1)),\ldots,( x_M,c_s(x_M))$ from the oracle $\mathbf{EX}$ and the computations of $f_N(x)$ for all $x\in\{x_1,\ldots,x_M\}$ by the quantum algorithm $\mathcal{A}$.
    The other parts, such as the Gaussian elimination, are deterministic, as shown in Algorithm~\ref{alg:algorithm of learning}.
    In loading the $M$ samples, based on Lemma~\ref{lemma:linear combination bound}, we require that the feature map $f_N(x)$ for the next $(M+1)$th sample $x$ from the same target distribution $\mathcal{D}_N$, which is to be evaluated after the learning from the $M$ samples, should be represented as a linear combination of those of the $M$ samples, $f_N(x_1),\ldots,f_N(x_M)$, with a high probability at least $1-\epsilon$; i.e., it should hold that
    \begin{equation}
    \label{eq:requirement_quantum_epsilon}
        \pr_{x\sim\mathcal{D}_N}\qty[f_N(x)=\sum_{m=1}^M\alpha_m f_N(x_m)]\geq1-\epsilon.
    \end{equation}
    Using Lemma~\ref{lemma:linear combination bound} with $y_1=f_N(x_1),\ldots,y_M=f_N(x_M)$, and $y=f_N(x)$, we see that, with $M$ given by~\eqref{eq:N_quantum_algorithm}, this requirement is fulfilled.
    Also, in the computations of $f_N$, we require that the probabilistic quantum algorithm $\mathcal{A}$ should simultaneously achieve
    \begin{align}
    \label{eq:requirement_quantum}
        \mathcal{A}(x_1)=f_N(x_1),\ldots,\mathcal{A}(x_M)=f_N(x_M),
    \end{align}
    with a high probability of at least $1-\delta$.
    For each $m\in\{1,\ldots,M\}$, due to~\eqref{eq:quantum_algorithm_function} and the union bound, we have $\mathcal{A}(x_m)=f_N(x_m)$ with a probability at least
    \begin{equation}
        1-(\mu+\nu);
    \end{equation}
    then, due to the union bound, the probability of having~\eqref{eq:requirement_quantum} simultaneously is at least
    \begin{equation}
        1-M(\mu+\nu).
    \end{equation}
    Thus, with $\mu$ chosen as~\eqref{eq:mu_quantum_algorithm} and $\nu$ as~\eqref{eq:nu_quantum_algorithm}, the requirement in~\eqref{eq:requirement_quantum} is fulfilled.
    As a whole, the requirement in~\eqref{eq:requirement_quantum_epsilon} is always satisfied for our choice of $M$, and the requirement in~\eqref{eq:requirement_quantum} is satisfied with a high probability at least $1-\delta$ for our choice of $\mu$ and $\nu$, which guarantees that the overall success probability of the learning algorithm is lower bounded by $1-\delta$.

    Given that the requirements in~\eqref{eq:requirement_quantum_epsilon} and~\eqref{eq:requirement_quantum} are fulfilled, the error in learning as in Definition~\ref{def: efficiently pac learnable} is bounded as follows.
    Under~\eqref{eq:requirement_quantum_epsilon} and~\eqref{eq:requirement_quantum}, for any $x$ satisfying
    \begin{equation}
    \label{eq:f_n_linear}
        f_N(x)=\sum_{m=1}^M\alpha_m f_N(x_m),
    \end{equation}
    the hypothesis $h_{\tilde{s}}$ in~\eqref{eq:hypothesis} parameterized by $\tilde{s}\in\mathbb{F}_2^D$ output by Algorithm~\ref{alg:algorithm of learning} can correctly classify $x$ as
    \begin{align}
        h_{\tilde{s}}(x)&=f_N(x)\cdot \tilde{s} \\ 
        \label{eq:1}
        &=\sum_{m=1}^M \alpha_mf_N(x_m)\cdot \tilde{s} \\
        \label{eq:2}
        &=\sum_{m=1}^M \alpha_m\mathcal{A}(x_m)\cdot \tilde{s} \\
        \label{eq:3}
        &=\sum_{m=1}^M \alpha_m c_s(x_m) \\
        &=\sum_{m=1}^M \alpha_mf_N(x_m)\cdot s \\
        &= f_N(x) \cdot s \\
        & = c_s(x),
    \end{align}
    where~\eqref{eq:1} follows from~\eqref{eq:f_n_linear},~\eqref{eq:2} from~\eqref{eq:requirement_quantum}, and~\eqref{eq:3} from~\eqref{eq:linear_equation}.
    Therefore, due to the requirement of~\eqref{eq:requirement_quantum_epsilon}, we have
    \begin{equation}
        \pr\qty[h(x)=c_s(x)]\geq 1-\epsilon;
    \end{equation}
    that is, the error in~\eqref{eq:error} is bounded by
    \begin{equation}
        \error(h)=\pr\qty[h(x)\neq c_s(x)]\leq\epsilon,
    \end{equation}
    as required for the learnability in Definition~\ref{def: efficiently pac learnable}.

    The runtime of Algorithm~\ref{alg:algorithm of learning} is dominated by the computations of $f_N$ by $\mathcal{A}$ and the Gaussian elimination.
    We first consider the runtime of computing $f_M$ for the $M$ samples $x_1,\ldots,x_M$.
    For each $x_m$ with $m\in\{1,\ldots,M\}$, the runtime of the quantum algorithm $\mathcal{A}$ for computing $f_N$ is given by $t_{\mathcal{A}}(x_m)$ in~\eqref{eq:runtime_t_A}; thus, the runtime of the $M$ calculations is
    \begin{equation}
    \sum_{m=1}^{M}t_{\mathcal{A}}(x_m)=O\qty(M\qty(\frac{N}{\mu\nu})^\alpha).
    \end{equation}
    In addition, the runtime of performing the Gaussian elimination to find a solution $\tilde{s}\in\mathbb{F}_2^D$ of the system of $M$ linear equations in~\eqref{eq:linear_equation} (with $D\leq M$ due to~\eqref{eq:N_quantum_algorithm}) is
    \begin{equation}
        O(M^3).
    \end{equation}
    In total, for $M$ in~\eqref{eq:N_quantum_algorithm}, $\mu$ in~\eqref{eq:mu_quantum_algorithm},  $\nu$ in~\eqref{eq:nu_quantum_algorithm}, and $D=O(\poly(N))=O(N^\beta)$ with some $\beta>0$, the overall runtime of Algorithm~\ref{alg:algorithm of learning} is upper bounded by
    \begin{align}
        &O\qty(M\qty(\frac{N}{\mu\nu})^\alpha)+O(M^3)\nonumber\\
        &=O\qty(\frac{N^{2\alpha \beta+\alpha+\beta}}{\delta^{2\alpha}\epsilon^{2\alpha+1}}+\frac{N^{3\beta}}{\epsilon^3})\\
        &=O\qty(\poly\qty(N,\frac{1}{\epsilon},\frac{1}{\delta})),
    \end{align}
    as required for efficient learnability in Definition~\ref{def: efficiently pac learnable}.
\end{proof}

Next, we describe our quantum algorithm for evaluating the hypotheses for our concept class.
Our quantum algorithm for evaluating a hypothesis $h_{\tilde{s}}$ in~\eqref{eq:hypothesis} with the learned parameter $\tilde{s}$ is given by Algorithm~\ref{alg:algorithm of evaluation}, where the hypothesis class is in~\eqref{eq:hypothesis_class}. 
In our case, the parameter $\tilde{s}$ serves as the $D$-bit representation of the hypothesis, corresponding to $\sigma_h$ in Definition~\ref{def: efficiently evaluatable} of the efficient evaluatability.
For any $\epsilon > 0$ and $\delta > 0$, our evaluation algorithm aims to achieve the efficient evaluation in Definition~\ref{def: efficiently evaluatable}; in particular,
the evaluation algorithm aims to output an estimate $\tilde{h}\in\{0,1\}$ of the hypothesis $h_{\tilde{s}}(x)$ for the input $x$ so as to satisfy
\begin{equation}
\label{eq:evaluation error}
    \pr_{x\sim\mathcal{D}_N}\left[\pr\qty[\tilde{h} = h_{\tilde{s}}(x)]\geq 1-\delta \right] \geq 1-\epsilon,
\end{equation}
where the inner probability is taken over the randomness of the evaluation algorithm.
To this goal, we set the internal parameters in Algorithm~\ref{alg:algorithm of evaluation} as
\begin{align}
    \mu = &\epsilon \label{eq:mu = varep}\\
    \nu = &\delta. \label{eq:nu = del}
\end{align}

In algorithm~\ref{alg:algorithm of evaluation}, an unseen input $x$ is initially given by sampling from the target distribution $\mathcal{D}_N$. 
Then, the algorithm probabilistically computes the quantumly advantageous function $f_N$ for the input $x$, using the same quantum algorithm $\mathcal{A}$ as that used in our learning algorithm, i.e., that in~\eqref{eq:quantum_algorithm_function} and~\eqref{eq:runtime_t_A}, yet with the parameters $\mu$ in~\eqref{eq:mu = varep} and $\nu$ in~\eqref{eq:nu = del}.
We let $\mathcal{A}(x)$ denote the output of $\mathcal{A}$ in~\eqref{eq:quantum_algorithm_function} for the input $x$, where we will omit $\mu$ and $\nu$ for simplicity of notation if it is obvious from the context.
Note that $\mathcal{A}$ may not be a deterministic algorithm; that is, we may have $\mathcal{A}(x) = f_N(x)$ only probabilistically, as shown in~\eqref{eq:quantum_algorithm_function}.
Finally, using the given parameter $\tilde{s}$ of the hypothesis $h_{\tilde{s}}$, the evaluation algorithm calculates the bitwise inner product of $\mathcal{A}(x)$ obtained from the above computation and $\tilde{s}$, so as to output
\begin{equation}
\label{eq:output of evaluation Algorithm}
    \tilde{h} \coloneqq \mathcal{A}(x)\cdot \tilde{s}.
\end{equation}

We now prove the quantumly efficient evaluatability of the hypothesis class $\mathcal{H}_N$ in~\eqref{eq:hypothesis_class} by Algorithm~\ref{alg:algorithm of evaluation} as follows.
\begin{theorem}[\label{thm: quantumly evaluatable}Quantumly efficient evaluatability]
    For any $N$, $D=O(\poly(N))$, any target distribution $\mathcal{D}_N$ over the $N$-bit input space $\mathcal{X}_N\subseteq\{0,1\}^N$, and any quantumly advantageous function $f_N:\{0,1\}^N\to\mathbb{F}_2^D$ under $\mathcal{D}_N$, the hypothesis class $\mathcal{H}_N$ in~\eqref{eq:hypothesis_class} with $f_N$, parameterized by $\tilde{s}\in\mathbb{F}^D$, is quantumly efficiently evaluatable by Algorithm~\ref{alg:algorithm of evaluation}.
\end{theorem}

\begin{proof}
    In the following, we will first discuss the success probability of our evaluation algorithm and then provide an upper bound of the runtime.

    The probabilistic part of Algorithm~\ref{alg:algorithm of evaluation} is confined solely to the computation of $f_N(x)$ by the quantum algorithm $\mathcal{A}$, and the other parts, such as the bitwise inner product, are deterministic.
    The requirement for this probabilistic part is that the quantum algorithm $\mathcal{A}$ should compute $f_N(x)$ correctly for a large fraction $1-\epsilon$ of the given input $x$ with high probability at least $1-\delta$, i.e.,
    \begin{equation}
        \pr_{x\sim\mathcal{D}_N}\left[\pr\qty[\mathcal{A}(x) = f_{N}(x)]\geq 1-\delta \right] \geq 1-\epsilon.
    \end{equation}
    Using $\mathcal{A}$ in~\eqref{eq:quantum_algorithm_function} with $\mu$ chosen as~\eqref{eq:mu = varep} and $\nu$ as \eqref{eq:nu = del}, we fulfill this requirement.
    Conditioned on having
    \begin{equation}
        \mathcal{A}(x) = f_{N}(x),
    \end{equation}
    the output $\tilde{h}$ in~\eqref{eq:output of evaluation Algorithm} becomes
    \begin{equation}
        \tilde{h}=\mathcal{A}(x)\cdot\tilde{s}=f_N(x)\cdot\tilde{s}=h_{\tilde{s}}(x).
    \end{equation}
    Consequently, Algorithm~\ref{alg:algorithm of evaluation} outputs $\tilde{h}$ satisfying
    \begin{equation}
        \pr_{x\sim\mathcal{D}_N}\left[\pr\qty[\tilde{h} = h_{\tilde{s}}(x)]\geq 1-\delta \right] \geq 1-\epsilon,
    \end{equation}
    as required for the evaulatability in Definition~\ref{def: efficiently evaluatable}.
    
    The runtime of Algorithm~\ref{alg:algorithm of evaluation} is dominated by the computation of $f_N$ by $\mathcal{A}$ and the bitwise inner product.
    We first consider the runtime $t_{\mathcal{A}}$ of $\mathcal{A}$ for the input $x$.
    We have the algorithm $\mathcal{A}$ satisfying \eqref{eq:runtime_t_A}.
    Accordingly, with $\mu$ chosen as~\eqref{eq:mu = varep} and $\nu$ as \eqref{eq:nu = del}, we have
    \begin{align}
        t_{\mathcal{A}}(x) &= O\left(\left(\frac{N}{\mu\nu}\right)^\alpha\right)\\
        &= O\left(\left(\frac{N}{\epsilon\delta}\right)^\alpha\right).
    \end{align}
    Also, the runtime of the bitwise inner product of vector in the $D$-dimensional vector space $\mathbb{F}_2^D$ over the finite field in \eqref{eq:output of evaluation Algorithm} is
    \begin{equation}
        O(D).
    \end{equation}
    Thus, for $\mu$ in~\eqref{eq:mu = varep}, $\nu$ in~\eqref{eq:nu = del}, and $D=O(\poly(N))=O(N^\beta)$ with some $\beta>0$, the overall runtime of Algorithm~\ref{alg:algorithm of evaluation} is upper bounded by
    \begin{align}
        &O\left(\left(\frac{N}{\epsilon\delta}\right)^\alpha\right) + O(D)\\
        &=O\left(\left(\frac{N}{\epsilon\delta}\right)^\alpha\right) + O\left(N^{\beta}\right)\\
        &=O\left(\poly\left(N,\frac{1}{\epsilon},\frac{1}{\delta}\right)\right),
    \end{align}
    as required for the efficient evaluatablity in Definition~\ref{def: efficiently evaluatable}. 
\end{proof}

\subsection{\label{sec:classical hardness}Provable hardness for any polynomial-time classical algorithm}
In this appendix, we prove the classical hardness of efficient learning and efficient evaluation for our concept class in Definition~\ref{def: concept class based on quantum advantage}.
Our proof is given by contradiction; that is, we will prove that, assuming that there exists a classically efficient evaluatable hypothesis class (Definition~\ref{def: efficiently evaluatable}) for classically efficient learnability (Definition~\ref{def: efficiently pac learnable}) of our concept class, one would be able to construct a polynomial-time classical algorithm to compute a quantumly advantageous function $f_N$ in Definition~\ref{def:quantum_advantage} using the polynomial-time classical algorithms for evaluating the hypotheses in this hypothesis class.
This classical algorithm is presented in Algorithm~\ref{alg:classical algorithm of evaluation}.
The rest of this appendix first describes this classical algorithm for the reduction of evaluating hypotheses to computing $f_N$ and then provides the full proof of the classical hardness.

To see the significance of our construction of this classical algorithm for the reduction, recall that it has been challenging to prove the classical hardness of learning without relying on discrete logarithms or integer factoring, which are solved by Shor's algorithms; by contrast, our proof of the classical hardness is applicable to any quantumly advantageous function beyond the scope of Shor's algorithms.
The technique of the proof by contradiction itself may be well established in the complexity theory and also used info showing the classical hardness of learning in the previous works~\cite{liu2021rigorous, servedio2004equivalences, kearns1990computational, kearns1994Jan,kearns1994introduction}.
However, the existing proofs of the classical hardness in these previous works essentially depend on a specific mathematical structure of discrete logarithms and integer factoring, so as to go through a cryptographic argument based on these computational problems.
To go beyond the realm of Shor's algorithms, novel techniques without relying on the existing cryptographic approach need to be developed.
By contrast, for our concept class with its feature space formulated as the space of bit strings, we prove the classical hardness based on any quantumly advantageous functions in Definition~\ref{def:quantum_advantage}, without depending on any specific quantum algorithm such as Shor's algorithms.

\begin{figure}[!t]
    \centering
    \begin{algorithm}[H]
        \caption{Classical algorithm for the reduction of evaluating the hypotheses for the concept class in Definition~\ref{def: concept class based on quantum advantage} to computing the quantumly advantageous function in Definition~\ref{def:quantum_advantage}}
        \label{alg:classical algorithm of evaluation}
        \begin{algorithmic}[1]
            \REQUIRE A new input $x\in\mathcal{X}_N$ sampled from the target distribution $\mathcal{D}_N$, an advice string $\alpha$ given by the representations $(\sigma_{h_{s_1}},\ldots,\sigma_{h_{s_D}})$ of $D$ hypotheses $h_{s_1},\ldots,h_{s_D}$ in~\eqref{eq:given hypotheses error} for the concept class $\mathcal{C}_N$ in Definition~\ref{def: concept class based on quantum advantage}, $\mu>0$, and $\nu>0$.
            \ENSURE An estimate $\tilde{f}$ of a quantumly advantageous function $\tilde{f}_N(x)$ for  $\mathcal{C}_N$ achieving~\eqref{eq:error of classical evaluation for function} and~\eqref{eq:runtime of classical evaluation for function}.
            \FOR{$d=1,\dots,D$}
            \STATE Perform the quantum algorithm $\mathcal{A}$ in \eqref{eq:classical evaluation for hypothesis error} and \eqref{eq:t_A_classical} for $x$, $\sigma_{h_{s_1}}$, $\epsilon$ in \eqref{eq:error of evaluation in classical}, and $\delta$ in~\eqref{eq:random of evaluation in classical}, to compute an estimate $\tilde{h}_{s_d}\in\{0,1\}$ of $h_{s_d}(x)$.
            \ENDFOR
            \RETURN $\tilde{f}=\qty(\tilde{h}_{s_1},\ldots,\tilde{h}_{s_1})^\top$ in~\eqref{eq:tilde_f}.
        \end{algorithmic}
    \end{algorithm}
\end{figure}

To show this, for any target distribution $\mathcal{D}_N$, any quantumly advantageous function $f_N:\{0,1\}^N\to\mathbb{F}_2^D$ under $\mathcal{D}_N$ in Definition~\ref{def:quantum_advantage} with $D=O(\poly(N))$, and our concept class $\mathcal{C}_N$ in Definition~\ref{def: concept class based on quantum advantage}, we assume that  $\mathcal{C}_N$ is classically efficiently learnable as in Definition~\ref{def: efficiently evaluatable} by a hypothesis class $\mathcal{H}_N$, and the hypothesis class $\mathcal{H}_N$ is classically efficiently evaluatable as in Definition~\ref{def: efficiently evaluatable}.
Under this assumption, we construct a polynomial-time classical algorithm for the reduction of the efficient evaluation of the hypotheses in $\mathcal{H}_N$ to the computation of $f_N$, as shown in Algorithm~\ref{alg:classical algorithm of evaluation}, which will lead to the contradiction.
Given an input $x$ drawn from $\mathcal{D}_N$ and an appropriate choice of a polynomial-length advice string $\alpha$ as in the definition of $\HeurC$ in~\eqref{eq:heurfbpp_poly},
the goal of Algorithm~\ref{alg:classical algorithm of evaluation} is, for all $0 < \mu< 1$ and $0< \nu <1$, to output an estimate $\tilde{f}\in\mathbb{F}_2^D$ of $f_N(x)$ satisfying
\begin{equation}
\label{eq:error of classical evaluation for function}
    \pr_{x\sim\mathcal{D}_N}[\pr[\tilde{f}=f_N(x)]\geq 1-\nu]\geq 1-\mu,
\end{equation}
within runtime 
\begin{equation}
\label{eq:runtime of classical evaluation for function}
    t_{\mathcal{A}}(x,\alpha,\mu,\nu) = O\left(\poly\left(N,\frac{1}{\mu},\frac{1}{\nu}\right)\right),
\end{equation}
where $\alpha$ will be chosen as the representations of $D$ hypotheses in $\mathcal{H}_N$ as described below.
To this goal, we set the internal parameters in Algorithm~\ref{alg:classical algorithm of evaluation} as
\begin{align}
\label{eq:error of learning in classical} 
    \epsilon_\mathrm{learn} &= \frac{\mu}{2D}, \\
\label{eq:random of learning in classical} 
    \delta_\mathrm{learn}&=\frac{1}{2},\\
    \label{eq:error of evaluation in classical}
    \epsilon_\mathrm{eval} &= \frac{\mu}{2D},  \\
    \label{eq:random of evaluation in classical}
    \delta_\mathrm{eval} &=  \frac{\nu}{D}.
\end{align}
Note that the choice of $\delta_\mathrm{learn}$ can be any constant between $0$ and $1$.

In Algorithm~\ref{alg:classical algorithm of evaluation}, an input $x$ drawn from the distribution $\mathcal{D}_N$ is initially given, and the representations of hypotheses for $D$ concepts in our concept class $\mathcal{C}_N$ are also initially given.
In particular, let
\begin{equation}
\label{eq:s_d}
    \{s_d\in\mathbb{F}_2^D\}_{d=1,\ldots,D}
\end{equation}
denote the standard basis of the $D$-dimensional vector space $\mathbb{F}_2^D$, where the $d$th element of the vector $s_d\in\mathbb{F}_2^D$ is $1$, and all the other elements of $s_d$ are $0$.
Then, under the assumption of the classically efficient learnability of $\mathcal{C}_N$, for each $s_d$ and all $0<\epsilon_\mathrm{learn},\delta_\mathrm{learn}<1$, there should exist a hypothesis $h_{s_d}$ such that
\begin{equation}
\label{eq:given hypotheses error}
    \pr_{x\sim\mathcal{D}_N}\qty[h_{s_d}(x)\neq c_{s_d}(x)]\leq \epsilon_\mathrm{learn},
\end{equation}
and the representation $\sigma_{h_{s_d}}$ of the hypothesis $h_{s_d}$ should be of polynomial length
\begin{equation}
\label{eq:length_representation}
    \size\qty(\sigma_{h_{s_d}})=O\qty(\qty(\frac{N}{\epsilon_\mathrm{learn}\delta_\mathrm{learn}})^\eta),
\end{equation}
where $\eta>0$ is an upper bound of the degree of the polynomial length.
Note that our proof of the hardness does not use the learning algorithm itself, but the assumption of classically efficient learnability is used to guarantee the existence of the hypotheses that approximate the concepts well and have polynomial-length representations, as in~\eqref{eq:given hypotheses error} and~\eqref{eq:length_representation}.
Furthermore, under the assumption of the classically efficient evaluatability of this hypothesis class, there should exist a classical (randomized) algorithm $\mathcal{A}$ such that for the representation $\sigma_{h_{s_d}}$ of each hypothesis $h_{s_d}$ with $s_d$ in~\eqref{eq:s_d}, and all $0<\epsilon_\mathrm{eval},\delta_\mathrm{eval}<1$, the algorithm $\mathcal{A}$ outputs an estimate $\tilde{h}_{s_d}\in\{0,1\}$ of $h_{s_d}(x)$ satisfying
\begin{align}
\label{eq:classical evaluation for hypothesis error}
    \pr_{x\sim\mathcal{D}_N}\left[\pr\qty[\tilde{h}_{s_d} = h_{s_d}(x)]\geq 1-\delta_\mathrm{eval} \right] \geq 1-\epsilon_\mathrm{eval}, 
\end{align}
within polynomial runtime for all $x$ in the support of $\mathcal{D}_N$
\begin{equation}
\label{eq:t_A_classical}
   t_{\mathcal{A}}\qty(x,\sigma_{h_{s_d}},\epsilon_\mathrm{eval},\delta_\mathrm{eval})=O\qty(\qty(\frac{N}{\epsilon_\mathrm{eval}\delta_\mathrm{eval}})^\gamma),
\end{equation}
where $\gamma>0$ is an upper bound of the degree of the polynomial runtime.

Under this assumption on the classically efficient learnability and the classically efficient evaluatability, Algorithm~\ref{alg:classical algorithm of evaluation} uses the classical evaluation algorithm $\mathcal{A}$ to compute each of the $D$ hypotheses $h_{s_1}(x),\ldots,h_{s_D}(x)$ for the input $x$, to obtain $\tilde{h}_{s_1},\ldots,\tilde{h}_{s_D}$.
Note that $\mathcal{A}$ may not be a deterministic algorithm, and thus, we may have $\tilde{h}_{s_d} = h_{s_d}(x)$ only probabilistically.
But if it holds that $\tilde{h}_{s_d}=h_{s_d}(x)=c_{s_d}(x)$, then $\tilde{h}_{s_d}$ is the $d$th bit of $f_N(x)\in\mathbb{F}_2^D$, as can be seen from~\eqref{eq:concept_class}.
Using this property of the vector space of bit strings, from the computed values $\tilde{h}_{s_1},\ldots,\tilde{h}_{s_D}\in\{0,1\}$, Algorithm~\ref{alg:classical algorithm of evaluation} outputs
\begin{equation}
\label{eq:tilde_f}
    \tilde{f}\coloneqq\begin{pmatrix}
    \tilde{h}_{s_1}\\
    \tilde{h}_{s_2}\\
    \vdots\\
    \tilde{h}_{s_D}
    \end{pmatrix}\in\mathbb{F}_2^D
\end{equation}
as an estimate of $f_N(x)$.

Using the reduction achieved by Algorithm~\ref{alg:classical algorithm of evaluation}, we prove that our concept class $\mathcal{C}_N$ is not classically efficiently learnable by any classically efficiently evaluatable hypothesis class.
We also note that, in previous works~\cite{liu2021rigorous, servedio2004equivalences, kearns1990computational, kearns1994Jan} of the classical hardness of learning tasks, efficient evaluatability was defined in terms of worst-case complexity; by contrast, motivated by the practical applicability as discussed in Appendix~\ref{sec:pac_learning}, our definition of efficient evaluatablity in Definition~\ref{def: efficiently evaluatable} is in terms of heuristic complexity.
Since the heuristic complexity classes include the corresponding worst-case complexity classes as discussed in Appendix~\ref{sec:quantum_advantage}, our proof of the classical hardness of our learning tasks for the heuristic complexity implies the more conventional classical hardness for the worst-case complexity as well.

\begin{theorem}[\label{thm: classical hardness for evaluation}Classical hardness]
    For any $N$, $D=O(\poly(N))$, any target distribution $\mathcal{D}_N$ over the $N$-bit input space $\mathcal{X}_N\subseteq\{0,1\}^N$, and any quantumly advantageous function $f_N:\{0,1\}^N\to\mathbb{F}_2^D$ under $\mathcal{D}_N$, the concept class $\mathcal{C}_N$ in Definition~\ref{def: concept class based on quantum advantage} with $f_N$ is not classically efficiently learnable by any classically efficiently evaluatable hypothesis class.
\end{theorem}

\begin{proof}
    We prove the statement by contradiction; i.e., we show that, under the assumption that $\mathcal{C}_N$ is classically efficiently learnable by some classically efficiently evaluatable hypothesis class, there should exist a classical algorithm (Algorithm~\ref{alg:classical algorithm of evaluation}) with a polynomial-length advice string $\alpha$ achieving \eqref{eq:error of classical evaluation for function} and \eqref{eq:runtime of classical evaluation for function} for the reduction to computing the quantum advantageous function $f_N$.
    In the following, we first analyze the length of $\alpha$.
    Then, we consider the success probability of our algorithm for the reduction.
    Finally, we discuss the runtime of our algorithm for the reduction.

    The length of the advice string $\alpha$ is bounded as follows.
    As the advice string $\alpha$, we use the representations
    \begin{equation}
        \alpha\coloneqq\qty(\sigma_{h_{s_1}},\ldots,\sigma_{h_{s_D}})
    \end{equation}
    of $D$ hypotheses in~\eqref{eq:given hypotheses error} and~\eqref{eq:length_representation}.
    Due to~\eqref{eq:length_representation},~\eqref{eq:error of learning in classical},~\eqref{eq:random of learning in classical}, and $D=O(\poly(N))=N^\beta$ for some $\beta>0$,
    the total length of $\alpha$ is
    \begin{align}
        \sum_{d=1}^D\size(\sigma_{h_{s_d}})&=O\qty(D\times\qty(\frac{N}{\epsilon_\mathrm{learn}\delta_\mathrm{learn}})^\eta)\\
        &=O\qty(\frac{N^{\beta\eta+\beta+\eta}}{\mu^\eta}),
    \end{align}
    as required for $\HeurC$ in~\eqref{eq:heurfbpp_poly}.

    Regarding the success probability of Algorithm~\ref{alg:classical algorithm of evaluation}, the probabilistic parts of the algorithm are the input $x$ from $\mathcal{D}_N$ inducing the error between the hypotheses $h_{s_1}(x),\dots,h_{s_D}(x)$ and the true concepts $c_{s_1}(x),\dots,c_{s_D}(x)$ in \eqref{eq:given hypotheses error}, and the computations of the estimates $\tilde{h}_{s_1},\ldots,\tilde{h}_{s_D}$ of the hypotheses $h_{s_1}(x),\dots,h_{s_D}(x)$ by the evaluation algorithm $\mathcal{A}$ in~\eqref{eq:classical evaluation for hypothesis error}.
    The other parts, such as the output of $\tilde{f}$ from $\tilde{h}_{s_1},\ldots,\tilde{h}_{s_D}$ in~\eqref{eq:tilde_f}, are deterministic, as shown in Algorithm~\ref{alg:classical algorithm of evaluation}.
    In Algorithm~\ref{alg:classical algorithm of evaluation}, we require that the hypotheses $h_{s_1}(x),\dots,h_{s_D}(x)$ simultaneously coincides with the true concepts $c_{s_1}(x),\dots,c_{s_D}(x)$, i.e.,
    \begin{equation}
    \label{eq:classical hypothesis evaluation simultaneously}
        h_{s_1}(x) = c_{s_1}(x),\dots,h_{s_D}(x) = c_{s_D}(x).
    \end{equation}
    With our choice of $\epsilon_\mathrm{learn}$ in~\eqref{eq:error of learning in classical}, due to~\eqref{eq:given hypotheses error} and the union bound, this requirement is fulfilled for a large fraction of $x$ at least
    \begin{equation}
    \label{eq:error_learn}
        1-D\epsilon_\mathrm{learn}=1-\frac{\mu}{2}.
    \end{equation}
    In addition, we require that the estimates $\tilde{h}_{s_1},\ldots,\tilde{h}_{s_D}$ simultaneously coincides with these hypotheses $h_{s_1}(x),\dots,h_{s_D}(x)$, i.e.
    \begin{equation}
    \label{eq:tilde_h_h}
        \tilde{h}_{s_1}=h_{s_1}(x),\dots,\tilde{h}_{s_D}=h_{s_D}(x).
    \end{equation}
    With our choice of $\epsilon_\mathrm{eval}$ in~\eqref{eq:error of evaluation in classical} and $\delta_\mathrm{eval}$ in~\eqref{eq:random of evaluation in classical}, due to~\eqref{eq:classical evaluation for hypothesis error} and the union bound, this requirement is fulfilled for a large fraction of $x$ at least
    \begin{equation}
    \label{eq:error_eval}
        1-D\epsilon_\mathrm{eval}=1-\frac{\mu}{2},
    \end{equation}
    with a high probability of at least
    \begin{equation}
    \label{eq:probability}
        1-D\delta_\mathrm{eval}=1-\nu.
    \end{equation}
    Given the requirements in~\eqref{eq:classical hypothesis evaluation simultaneously} and~\eqref{eq:tilde_h_h},
    due to~\eqref{eq:tilde_f},
    the output of Algorithm~\ref{alg:classical algorithm of evaluation} is
    \begin{equation}
        \tilde{f}=\begin{pmatrix}
    \tilde{h}_{s_1}\\
    \tilde{h}_{s_2}\\
    \vdots\\
    \tilde{h}_{s_D}
    \end{pmatrix}=\begin{pmatrix}
    c_{s_1}(x)\\
    c_{s_2}(x)\\
    \vdots\\
    c_{s_D}(x)
    \end{pmatrix}=f_N(x),
    \end{equation}
    where the last equality follows from~\eqref{eq:concept_class} since $\{s_d\}_{d}$ is the standard basis of the $D$-dimensioanl vector space $\mathbb{F}_2^D$. 
    Consequently, due to~\eqref{eq:error_learn},~\eqref{eq:error_eval},~\eqref{eq:probability}, and the union bound, the requirements in~\eqref{eq:classical hypothesis evaluation simultaneously} and~\eqref{eq:tilde_h_h} are simultaneously fulfilled 
    for a large fraction of $x$ at least
    \begin{equation}
        1-\mu,
    \end{equation}
    with a high probability of at least
    \begin{equation}
        1-\nu,
    \end{equation}
    which yields the success probability of our algorithm as required for $\HeurC$ in~\eqref{eq:heurfbpp_poly}.

    The runtime of Algorithm~\ref{alg:classical algorithm of evaluation} is determined by the evaluations of the $D$ hypotheses and the bitwise inner product.
    For any $x$ and every $d\in{1,\dots,D}$, the runtime of the classical algorithm $\mathcal{A}$ for computing $h_{s_d}$ is given by
    \begin{equation}
        t_{\mathcal{A}}\qty(x,\sigma_{h_{s_d}},\epsilon_\mathrm{eval},\delta_\mathrm{eval})=O\qty(\qty(\frac{N}{\epsilon_\mathrm{eval}\delta_\mathrm{eval}})^\gamma),
    \end{equation}
    as shown in \eqref{eq:t_A_classical}.
    Thus, the runtime of the $D$ evaluations is
    \begin{align}
    \label{eq:classical runtime_1}
    \sum_{d=1}^{D}t_{\mathcal{A}}\qty(x,\sigma_{h_{s_d}},\epsilon_\mathrm{eval},\delta_\mathrm{eval})=O\qty(D\qty(\frac{N}{\epsilon_\mathrm{eval}\delta_\mathrm{eval}})^\gamma).
    \end{align}
    In addition, the runtime of the output of the $D$-dimensional vector $\tilde{f}$ in~\eqref{eq:tilde_f} is
    \begin{equation}
    \label{eq:classical runtime_2}
        O(D).
    \end{equation}
    Due to~\eqref{eq:classical runtime_1} and~\eqref{eq:classical runtime_2}, for $\epsilon_\mathrm{eval}$ in~\eqref{eq:error of evaluation in classical}, $\delta_\mathrm{eval}$ in~\eqref{eq:random of evaluation in classical}, and $D=O(N^\beta)$ with some constant $\beta>0$, the overall runtime of Algorithm~\ref{alg:classical algorithm of evaluation} is upper bounded by
    \begin{equation}
    \label{eq:runtime of algorithm 3}
        \begin{aligned}
            &O\left(D\qty(\frac{N}{\epsilon_\mathrm{eval}\delta_\mathrm{eval}})^\gamma\right) + O(D)\\
            &O\left(N^\beta\qty(\frac{N^\gamma N^{\beta\gamma}N^{\beta\gamma}}{\mu^\gamma\nu^\gamma})\right)\\
            &=O\left(\frac{N^{2\beta\gamma+\beta+\gamma}}{\mu^\gamma\nu^\gamma}\right)\\
            &=O\left(\poly\left(N,\frac{1}{\mu},\frac{1}{\nu}\right)\right),
        \end{aligned}
    \end{equation}
    as required for $\HeurC$ in~\eqref{eq:heurfbpp_poly}.
    
    Consequently, under the assumption that $\mathcal{C}_N$ is classically efficiently learnable by some classically efficiently evaluatable hypothesis class, one would be able to construct Algorithm~\ref{alg:classical algorithm of evaluation} achieving~\eqref{eq:error of classical evaluation for function} and~\eqref{eq:runtime of classical evaluation for function}; that is, the problem $\{(f_N,\mathcal{D}_N)\}$ would be in $\HeurC$.
    This contradicts Definition~\ref{def:quantum_advantage} of the quantum advantageous function $f_N$.
\end{proof}

\section{\label{sec: Learning advantage without conditions on the hypothetical class}Data-preparation protocols for demonstrating advantage of QML from general computational advantages}

In this appendix, we propose protocols for preparing the sample data for our learning tasks studied in Appendix~\ref{sec:advantage_QML} so as to demonstrate the advantage of QML using our learning tasks.
A nontrivial part of our analysis of this data preparation is that the sample data can be prepared only probabilistically in our general setting; after all, the quantumly advantageous functions used for our concept class in Definition~\ref{def: concept class based on quantum advantage} are defined for probabilistic algorithms and heuristic complexity classes in general.
Nevertheless, we provide feasible conditions for the correct data preparation.
The rest of this appendix is organized as follows.
In Appendix~\ref{sec:framework}, we provide a two-party setup for demonstrating the advantage of QML with one party preparing the data and the other learning from the data.
In Appendix~\ref{sec:quantum_data}, we describe a protocol using quantum computation to prepare the correct sample data with a high success probability for the demonstration.
In Appendix~\ref{sec:classical_data}, we describe another protocol using classical computation to prepare the sample data with a high success probability for the demonstration, in special cases where the quantumly advantageous function is constructed based on a class of one-way permutation that is hard to invert by a polynomial-time classical algorithm but can be inverted by a polynomial-time quantum algorithm.

\subsection{\label{sec:framework}Setup for demonstrating advantage of QML from general computational advantages}
This appendix provides a setup for demonstrating the advantage of QML\@.
In other words, we propose a learning setting including data preparation.

In our setup, we consider two parties; a party $A$ is in charge of data preparation, and the other party $B$ receives sample data from $A$ to perform learning.
The party $A$ uses either quantum or classical computers to prepare the data
while $B$ does not know how $A$ has prepared the data.
The data should be prepared in such a way that $B$ can achieve the learning if $B$ uses the quantum learning algorithm in Appendix~\ref{sec:quantum_algorithm} but cannot if $B$ is limited to any polynomial-time classical learning method as in Appendix~\ref{sec:classical hardness}.
See also the main text for an illustration of the setup.

The overall protocol for $A$ and $B$ demonstrating the advantage of QML in this setup is as follows.
First, two parties $A$ and $B$ are given the problem size $N$, the concept class $\mathcal{C}_N = \{c_s\}_{s\in\mathbb{F}_2^D}$ specified by the quantumly advantageous function $f_N:\mathbb{F}_2^N\to\mathbb{F}_2^D$ under a target distribution $\mathcal{D}_N$ in Definition~\ref{def: concept class based on quantum advantage}, the error parameter $\epsilon$ and the confidence parameter $\delta$, where $D = O(\poly(N))$.
Note that the description of $\mathcal{D}_N$ may be unknown to $A$ and $B$ throughout the protocol, but $A$ has access to (an oracle to load) an $O(\poly(N,1/\epsilon,1/\delta))$ amount of the inputs $x$ sampled from $\mathcal{D}_N$ within a unit time per loading each input.
Given these parameters, based on~\eqref{eq:N_quantum_algorithm}, $B$ determines the number of samples for learning as
\begin{equation}
\label{eq:M}
    M = \left\lceil\frac{D}{\epsilon}-1\right\rceil,
\end{equation}
where $\lceil{}\cdots{}\rceil$ is the ceiling function,
and send $M$ to $A$.
Then, $A$ decides the parameter $s$ of the target concept $c_s$ arbitrarily and keeps $s$ as $A$'s secret.
For $M$ and $s$, the task of $A$ is to correctly prepare the $M$ sample data
\begin{equation}
\label{eq:data}
    \{(x_m,c_s(x_m))\}_{m=1}^M,
\end{equation}
using a quantum or classical computer.
After preparing the $M$ data in~\eqref{eq:data}, $A$ sends the data to $B$.
Using the given sample data, the task of $B$ is to find a parameter $\tilde{s}\in\mathbb{F}_2^D$ and make a prediction for new input $x$ drawn from $\mathcal{D}_N$ by the hypothesis $h_{\tilde{s}}(x) = f(x) \cdot \tilde{s}$ so that the error should satisfy
\begin{align}
    \pr_{x\sim\mathcal{D}_N}[h_{\tilde{s}}(x) \neq c_s(x)] \leq \epsilon
\end{align}
with a high probability of at least $1-\delta$.

Using Algorithm~\ref{alg:algorithm of learning} and Algorithm~\ref{alg:algorithm of evaluation}, $B$ can achieve this task with quantum computation within a polynomial time
\begin{equation}
    O(\poly(N,1/\epsilon,1/\delta),
\end{equation}
and our analysis in Appendix~\ref{sec:classical hardness} shows that $B$ cannot achieve this task with any polynomial-time classical method.
In the following appendices, we will construct $A$'s algorithms for preparing the data in~\eqref{eq:data} with a high probability of at least $1-\delta$ within a polynomial time
\begin{equation}
    O(\poly(N,1/\epsilon,1/\delta)).
\end{equation}
With a sufficient amount of correct data, one can conduct the learning and test the learned hypothesis.
Thus, with $A$'s data-preparation algorithm and $B$'s learning and evaluation algorithms, our protocol in the above setup can demonstrate the advantage of QML\@.

\subsection{\label{sec:quantum_data}Quantum algorithm for preparing data}
In this appendix, we show how to prepare the sample data in~\eqref{eq:data} for the concept class in Definition~\ref{def: concept class based on quantum advantage} based on a quantumly advantageous function $f_N$ in Definition~\ref{def:quantum_advantage}, using a quantum algorithm.

\begin{figure}[!t]
    \centering
    \begin{algorithm}[H]
        \caption{Quantum algorithm for data preparation}
        \label{alg:data_prep_quantum}
        \begin{algorithmic}[1]
            \REQUIRE The problem size $N$, the concept class $\mathcal{C}_N$ in Definition~\ref{def: concept class based on quantum advantage} with a quantumly advantageous function $f_N$, $\epsilon>0$, $\delta>0$, the number $M=O(\poly(N,1/\epsilon,1/\delta))$ of samples to be prepared (e.g., given by~\eqref{eq:M}), the true parameter $s$ of the target concept, inputs sampled from a target distribution $\mathcal{D}_N$ to be loaded from an oracle.
            \ENSURE An estimate $\{(x_m,\tilde{c}_m)\}_{m=1}^{M}$ of the $M$ sample data $\{(x_m,c_s(x_m))\}_{m=1}^{M}$ satisfying~\eqref{eq:goal_data_prep} with a high probability at least $1-\delta$.
            \FOR{$m=1,\dots,M$}
            \STATE Load an input $x_m$ sampled from the target distribution $\mathcal{D}_N$ (with access to the oracle).
            \STATE Perform the quantum algorithm $\mathcal{A}$ in~\eqref{eq:quantum_algorithm_function} for computing $f_N$ for the input $x_m$ with the parameters $\mu$ and $\nu$ in~\eqref{eq:mu_data_prep} and~\eqref{eq:nu_data_prep}, respectively, to obtain $\mathcal{A}(x_m)$.
            \STATE Compute $\tilde{c}_m=\mathcal{A}(x_m)\cdot s$ in~\eqref{eq:output_data}.
            \ENDFOR
            \RETURN $\{(x_m,\tilde{c}_m)\}_{m=1}^{M}$.
        \end{algorithmic}
    \end{algorithm}
\end{figure}

The data-preparation algorithm is shown in Algorithm~\ref{alg:data_prep_quantum}.
As described in Appendix~\ref{sec:framework},
given the problem size $N$, the concept class $\mathcal{C}_N$ in Definition~\ref{def: concept class based on quantum advantage} with a quantumly advantageous function $f_N$, the error parameter $\epsilon$, the significance parameter $\delta$, the number $M=O(\poly(N,1/\epsilon,1/\delta))$ of samples to be prepared (e.g., given by~\eqref{eq:M}, yet we here describe the algorithm for general $M$), and the true parameter $s$ of the target concept, we assume that Algorithm~\ref{alg:data_prep_quantum} has access to (an oracle to load) an $O(\poly(N,1/\epsilon,1/\delta))$ amount of the inputs $x$ sampled from $\mathcal{D}_N$ within a unit time per loading each input.
Then, the goal of Algorithm~\ref{alg:data_prep_quantum} is to output an estimate $\{(x_m,\tilde{c}_m)\}_{m=1}^{M}$ of the data $\{(x_m,c_s(x_m))\}_{m=1}^{M}$ in~\eqref{eq:data} with $x_m$ drawn from the target distribution $\mathcal{D}_N$, so that it should hold with a high probability at least $1-\delta$ that, for all $m$,
\begin{equation}
\label{eq:goal_data_prep}
    \tilde{c}_m=c_s(x_m).
\end{equation}
Note that the error parameter $\epsilon$ is not explicitly relevant to Algorithm~\ref{alg:data_prep_quantum} except for the possibility of $M$ depending on $\epsilon$.
To this goal, we set the internal parameters in Algorithm~\ref{alg:data_prep_quantum} as
\begin{align}
\label{eq:mu_data_prep}
    \mu&=\frac{\delta}{2M},\\
\label{eq:nu_data_prep}
    \nu&=\frac{\delta}{2M}.
\end{align}
In Algorithm~\ref{alg:data_prep_quantum}, for each $m=1,\ldots,M$, we start with sampling $x_m$ from the target distribution $\mathcal{D}_N$.
Then, we use the quantum algorithm $\mathcal{A}$ in~\eqref{eq:quantum_algorithm_function} to compute the quantumly advantageous function $f_N$ with $\mu$ in~\eqref{eq:mu_data_prep} and $\nu$ in~\eqref{eq:nu_data_prep}, to obtain $\mathcal{A}(x_m)$ within a polynomial runtime in~\eqref{eq:runtime_t_A}, where $\mu$ and $\nu$ in~\eqref{eq:quantum_algorithm_function} may be omitted for simplicity of the presentation if obvious from the context.
Finally, Algorithm~\ref{alg:data_prep_quantum} computes an estimate $\tilde{c}_m$ of $c_s(x_m)$ by
\begin{equation}
\label{eq:output_data}
    \tilde{c}_m\coloneqq\mathcal{A}(x_m)\cdot s,
\end{equation}
in accordance with the definition of $c_s$ in~\eqref{eq:concept_class}.
After performing these computations for all $m$, the algorithm outputs
\begin{equation}
    \{(x_m,\tilde{c}_m)\}_{m=1}^M
\end{equation}
as an estimate of the data $\{(x_m,c_s(x_m))\}_{m=1}^{M}$ in~\eqref{eq:data}.

The following theorem shows that this algorithm prepares the data  in~\eqref{eq:data} correctly with a high probability $1-\delta$ within a polynomial time.

\begin{theorem}[The polynomial-time data preparation with a quantum algorithm]
Given any $N$, $\epsilon$, and $\delta$, for any $M=O(\poly(N,1/\epsilon,1/\delta))$,
Algorithm~\ref{alg:data_prep_quantum} outputs the $M$ sample data $\{(x_m,c_s(x_m))\}_{m=1}^{M}$ in~\eqref{eq:data} with a high probability at least $1-\delta$ within a polynomial time
\begin{equation}
    O\qty(\poly(N,1/\epsilon,1/\delta)).
\end{equation}
\end{theorem}

\begin{proof}
    We will first discuss the success probability of Algorithm~\ref{alg:data_prep_quantum} and then provide an upper bound of the runtime.

    The probabilistic parts of Algorithm~\ref{alg:data_prep_quantum} are the computations of $f_N$ by the quantum algorithm $\mathcal{A}$.
    We require that for all $m\in\{1,\ldots,M\}$, the output $\mathcal{A}(x_m)$ of this quantum algorithm should coincide with $f_N(x_m)$ simultaneously, i.e.
    \begin{equation}
        \mathcal{A}(x_1)=f_N(x_1),\ldots,\mathcal{A}(x_M)=f_N(x_M).
    \end{equation}
    Given this requirement, the output of Algorithm~\ref{alg:data_prep_quantum} coincides with the data in~\eqref{eq:data}; that is, $\tilde{c}_m$ in~\eqref{eq:output_data} satisfies
    \begin{equation}
        \tilde{c}_m=\mathcal{A}(x_m)\cdot s=f_N(x_m)\cdot s=c_s(x_m),
    \end{equation}
    by definition of $c_s$ in~\eqref{eq:concept_class}.
    With our choice of $\mu$ in~\eqref{eq:mu_data_prep} and $\nu$ in~\eqref{eq:nu_data_prep}, due to the union bound, this requirement is fulfilled with a high probability at least
    \begin{equation}
        1-M(\mu+\nu)=1-\delta,
    \end{equation}
    which yields the desired success probability.
    
    The runtime of Algorithm~\ref{alg:data_prep_quantum} is determined by the computations of $f_N$ and the bitwise inner product.
    Due to \eqref{eq:runtime_t_A} with the choice of $\mu$ in~\eqref{eq:mu_data_prep} and $\nu$ in~\eqref{eq:nu_data_prep}, for $D = O(N^\beta)$ with $\beta>0$ and $M=O((\frac{N}{\epsilon\delta})^\xi)$ with $\xi>0$, we have the overall runtime
    \begin{align}
            &O\left(M\left(\left(\frac{N}{\mu\nu}\right)^\alpha+D\right)\right)\\
            &= 
            O\left(\frac{N^{\alpha+2\alpha\xi+\xi}}{\epsilon^{(2\alpha+1)\xi}\delta^{2\alpha\xi+2\alpha+\xi}}+\frac{N^{\beta+\xi}}{\epsilon^\xi\delta^\xi}\right)\\
            &= O\left(\poly\left(N,1/\epsilon,1/\delta\right)\right),
    \end{align}
    which yields the conclusion.
\end{proof}

\subsection{\label{sec:classical_data}Classical algorithm for preparing data based on classically one-way permutation}
In this appendix, we show how to prepare the sample data in \eqref{eq:data} using a classical algorithm, for a concept class derived by replacing the quantumly advantageous function used in Definition~\ref{def: concept class based on quantum advantage} with an inverse of a classically one-way permutation introduced in the following.

We define the classically one-way permutation to derive the concept class for preparing the data with the classical algorithm.

\begin{definition}[\label{def:one-way permutation}Classically one-way permutation]
    For $N$, let $f^{\mathrm{OWP}}_N:\{0,1\}^N\to\{0,1\}^N$ be a permutation (i.e., an $N$-bit one-to-one function), where we may write $\mathbb{F}_2^N=\{0,1\}^N$. 
    We write
    \begin{equation}
    \label{eq:y}
        x=f^{\mathrm{OWP}}_N(y),
    \end{equation}
    where sampling $x$ from a probability distribution $\mathcal{D}_N$ with computing $y=f^{\mathrm{OWP}^{-1}}_N(x)$ is equivalent to sampling $y$ from a probability distribution $\mathcal{D}_N^{Y}$ with computing~\eqref{eq:y} under the condition that
    \begin{equation}
    \label{eq:D_out}
        \mathcal{D}_N=f^{\mathrm{OWP}}_N(\mathcal{D}_N^{Y}).
    \end{equation}
    We say that a permutation $f^{\mathrm{OWP}}_N$ is a classically one-way permutation $f^{\mathrm{OWP}}_N$ under $\mathcal{D}_N$ if the relation $R \coloneqq \{R_N\}_{N\in\mathbb{N}}$ with $R_N \coloneqq \{\left(y,f^{\mathrm{OWP}}_N(y)\right)\}_{y\in\{0,1\}^N}$ is in $\textsf{FP}$, and the distributional function problem $\{(f_N^{\mathrm{OWP}^{-1}},\mathcal{D}_N)\}_{N\in\mathbb{N}}$ is in $\HeurQ$ but not in $\HeurC$.
\end{definition}

We then introduce the following concept class by replacing the quantumly advantageous function in the concept class of Definition~\ref{def: concept class based on quantum advantage} with the classically one-way permutation in Definition~\ref{def:one-way permutation}.

\begin{definition}[\label{def: concept class based on oneway perm}Concept class based on classically one-way permutation]
    For any $N$, any probability distribution $\mathcal{D}_N$ over $\mathbb{F}_2^N$, and any classically one-way permutation $f_N^{\mathrm{OWP}}:\mathbb{F}_2^N\to\{0,1\}^N$ under $\mathcal{D}_N$ in Definition~\ref{def:one-way permutation}, we define a concept class $\mathcal{C}_N^{\mathrm{OWP}}$ over the input space $\mathcal{X}_N$ as $\mathcal{C}_N^{\mathrm{OWP}}\coloneqq\{c_s\}_{s\in{\mathbb{F}_2^N}}$, where $\mathcal{X}_N$ is the support of $\mathcal{D}_N$ in~\eqref{eq:D_out}, and $c_s$ is a concept given by
    \begin{equation}
    \label{eq: concept based on one-way}
    c_s(x)\coloneqq f_N^{\mathrm{OWP}^{-1}}(x)\cdot s\in\mathbb{F}_2=\{0,1\},
\end{equation}
with $f_N^{\mathrm{OWP}^{-1}}(x)\cdot s$ denoting a bitwise inner product in the vector space $\mathbb{F}_2^N$ over the finite field. 
\end{definition}

By definition, the inverse $f_N^{{\mathrm{OWP}^{-1}}}$ of the classically one-way permutation $f_N^{\mathrm{OWP}}$ in Definition~\ref{def: concept class based on oneway perm} is a special case of the quantumly advantageous functions in Definition~\ref{def: concept class based on quantum advantage}. 
Therefore, the quantum efficient learnability, the quantum efficient evaluatability, and the classical hardness for this concept class follow from the same argument as Appendix~\ref{sec:advantage_QML}.
Note that particular variants of classically one-way permutations $f_N^{\mathrm{OWP}}$ that can be inverted by Shor's algorithms are used in the previous work on the advantage of QML~\cite{servedio2004equivalences, liu2021rigorous}.
Since $f_N^{\mathrm{OWP}}$ is a permutation, if the target distribution $\mathcal{D}_N$ is uniform, then Algorithm~\ref{alg:data_prep_classical} simply samples from the uniform distribution, which is assumed in Refs.~\cite{servedio2004equivalences, liu2021rigorous}.
By contrast, our analysis does not assume the uniform distribution, generalizing the settings in Refs.~\cite{servedio2004equivalences, liu2021rigorous}.
And even more importantly, the concept class in Definition~\ref{def: concept class based on oneway perm} does not depend on specific cryptographic techniques for the classically one-way permutation $f_N^{\mathrm{OWP}}$ such as those invertible by Shor's algorithms, in the same way as the concept class in Definition~\ref{def: concept class based on quantum advantage} without depending on the specific mathematical structure of quantumly advantageous functions.

\begin{figure}[!t]
    \centering
    \begin{algorithm}[H]
        \caption{Classical algorithm for data preparation with classically one-way function}
        \label{alg:data_prep_classical}
        \begin{algorithmic}[1]
            \REQUIRE The problem size $N$, the concept class $\mathcal{C}_N$ in Definition~\ref{def: concept class based on oneway perm} with a classically one-way permutation $f_N^\mathrm{OWP}$, $\epsilon>0$, $\delta>0$, the number $M=O(\poly(N,1/\epsilon,1/\delta))$ of samples to be prepared (e.g., given by~\eqref{eq:M}), the true parameter $s$ of the target concept, and parameters $y$ sampled from a probability distribution $\mathcal{D}_N^{Y}$ in~\eqref{eq:D_out} to be loaded from an oracle.
            \ENSURE The $M$ sample data $\{(x_m,c_s(x_m))\}_{m=1}^{M}$ in \eqref{eq:data_classi}.
            \FOR{$m=1,\dots,M$}
            \STATE Load a parameter $y_m$ sampled from the distribution $\mathcal{D}_N^Y$ (with access to the oracle).
            \STATE Perform the deterministic classical algorithm in \eqref{eq:x_m} to compute $f_N^{\mathrm{OWP}}$ for $y_m$, to obtain $x_m =f_N^{\mathrm{OWP}}(y_m)$.
            \STATE Compute $c_s(x_m)=y_m\cdot s$ in~\eqref{eq:c_s_classical}.
            \ENDFOR
            \RETURN $\{(x_m,c_s(x_m))\}_{m=1}^{M}$.
        \end{algorithmic}
    \end{algorithm}
\end{figure}

In the following, based on the setup described in Appendix~\ref{sec:framework}, we modify the protocol in such a way that the concept class is replaced with the above concept class based on a classically one-way permutation, and the party $A$ has access to (an oracle to load) an $O(\poly(N,1/\epsilon,1/\delta))$ amount of the parameters $y$ in~\eqref{eq:y} sampled from $\mathcal{D}_N^Y$ in~\eqref{eq:D_out}, in place of loading $x$, within a unit time per loading each $y$.

The classical data-preparation algorithm is shown in Algorithm~\ref{alg:data_prep_classical}.
Given the problem size $N$, the concept class $\mathcal{C}_N$ in Definition~\ref{def: concept class based on oneway perm} with a classically one-way permutation $f^{\mathrm{OWP}}_N$ under the probability distribution $\mathcal{D}_N$ in Definition~\ref{def:one-way permutation}, the error parameter $\epsilon$, the significance parameter $\delta$, the number $M=O(\poly(N,1/\epsilon,1/\delta))$ of samples to be prepared (e.g., given by~\eqref{eq:M}, yet we here describe the algorithm for general $M$), the true parameter $s$ of the target concept, and the parameters $y$ to be loaded as assumed above, the goal of Algorithm~\ref{alg:data_prep_classical} is to output $M$ pairs of data
\begin{equation}
\label{eq:data_classi}
    \{(x_m,c_s(x_m))\}_{m=1}^{M},
\end{equation}
with each $x_m$ drawn from the distribution $\mathcal{D}_N$, where \eqref{eq:data_classi} is a variant of \eqref{eq:data} up to the change of the concept class to~\eqref{eq: concept based on one-way}.
Note that the error parameter $\epsilon$ and the significance parameter $\delta$ are not explicitly relevant to Algorithm~\ref{alg:data_prep_classical} except for the possibility of $M$ depending on $\epsilon$ and $\delta$.
In Algorithm~\ref{alg:data_prep_classical}, for each $m = 1,\dots, M$, we start with loading $y_m$ sampled from the distribution $\mathcal{D}_N^Y$.
Then, the algorithm computes the classically one-way permutation $f_N^\mathrm{OWP}$ for $y_m$.
By definition of the classically one-way permutation $\{f_N^{\mathrm{OWP}}\}_{N\in\mathbb{N}}\in \textsf{FP}$, we have a polynomial-time deterministic classical algorithm $\mathcal{A}$ to compute
\begin{equation}
\label{eq:x_m}
    x_m=\mathcal{A}(y_m)=f_N^{\mathrm{OWP}}(y_m)
\end{equation}
within runtime
\begin{equation}
\label{eq:classical_alg_owp}
    t_\mathcal{A}(y_m) = O\left(N^\zeta\right),
\end{equation}
where $\zeta>0$ is an upper bound of the degree of the polynomial runtime.
Finally, Algorithm~\ref{alg:data_prep_classical} computes $c_s(x_m)$ by
\begin{equation}
\label{eq:c_s_classical}
    c_s(x_m) = y_m\cdot s,
\end{equation}
following the definition of $c_s$ in~\eqref{eq: concept based on one-way}.
After performing these computations for all $m$, Algorithm~\ref{alg:data_prep_classical} outputs the data in \eqref{eq:data_classi}, i.e.,
\begin{equation}
    \{(x_m,c_s(x_m))\}_{m=1}^M.
\end{equation}

The following theorem shows that Algorithm~\ref{alg:data_prep_classical} prepares the data in \eqref{eq:data_classi} correctly within a polynomial time.
\begin{theorem}[The polynomial-time data preparation with a classical algorithm based on classically one-way functions]
Given any $N$, $\epsilon$, and $\delta$, for any $M = O(\poly(N,1/\epsilon,1/\delta))$, Algorithm~\ref{alg:data_prep_classical} outputs the $M$ sample data $\{(x_m,c_s(x_m))\}_{m=1}^M$ in \eqref{eq:data_classi} within a polynomial time 
    \begin{equation}
    O(\poly(N,1/\epsilon,1/\delta)).
    \end{equation}
\end{theorem}
\begin{proof}[Proof]
    Algorithm~\ref{alg:data_prep_classical} is a deterministic algorithm and has no error; thus, it suffices to discuss the runtime.
    The runtime of Algorithm~\ref{alg:data_prep_classical} is determined by computing the classically one-way permutation $f_N^{\mathrm{OWP}}$ in Definition~\ref{def:one-way permutation} for $M$ inputs $y_1,\dots y_M$ and bitwise inner product.
    Due to~\eqref{eq:classical_alg_owp}, for $M = O((\frac{N}{\epsilon\delta})^\xi)$ with $\xi > 0$, we have the overall runtime
    \begin{align}
        &O\left(M\left(N^\zeta + N\right)\right)\\
        &= O\left(\frac{N^{\zeta+\xi}}{\epsilon^\xi\delta^\xi}\right)\\
        &= O\left(\poly\left(\frac{N}{\epsilon\delta}\right)\right),
    \end{align}
    which yields the conclusion.
\end{proof}

\bibliography{main.bib}

\end{document}